\documentclass[acmsmall,nonacm]{acmart} 
\usepackage{amsthm}
\usepackage[capitalise]{cleveref}
\usepackage{algorithm}
\usepackage{graphicx}
\usepackage{algpseudocode}
\newtheorem{theorem}{Theorem}[section]
\newtheorem*{*theorem}{Theorem}
\newtheorem{corollary}{Corollary}[theorem]
\newtheorem{lemma}[theorem]{Lemma}
\newtheorem{observation}[theorem]{Observation}

\newtheorem{claim}[theorem]{Claim}

\newenvironment{claimproof}{\paragraph{Proof:}}{\hfill$\square$}

\newtheorem{definition}{Definition}[section]
\newcommand{\StrongTrueStatelessness}{\emph{Strict Statelessness}}
\newcommand{\StrictStatelessness}{\emph{Strict Statelessness}}
\newcommand{\Obliviousness}{\emph{Obliviousness}}
\newcommand{\Determinism}{\emph{Determinism}}
\newcommand{\Bandwidth}{\emph{Unit Bandwidth}}

\AtBeginDocument{%
  }

\setcopyright{acmlicensed}
\copyrightyear{2018}
\acmYear{2018}
\acmDOI{XXXXXXX.XXXXXXX}
\acmConference[Conference acronym 'XX]{Make sure to enter the correct
  conference title from your rights confirmation emai}{June 03--05,
  2018}{Woodstock, NY}
\acmISBN{978-1-4503-XXXX-X/18/06}

\usepackage[colorinlistoftodos,textsize=small]{todonotes}

\newcommand{\anew}[1]{#1}

\begin{document}


\title{Amnesiac Flooding: Easy to break, hard to escape}

\author{Henry Austin}
\email{henry.b.austin@durham.ac.uk}
\orcid{}
\author{Maximilien Gadouleau}
\email{m.r.gadouleau@durham.ac.uk}
\orcid{0000-0003-4701-738X}
\author{George B. Mertzios}
\authornote{Supported by the EPSRC grant EP/P020372/1.}
\email{george.mertzios@durham.ac.uk}
\orcid{0000-0001-7182-585X}
\author{Amitabh Trehan}
\email{amitabh.trehan@durham.ac.uk}
\authornote{Supported by the EPSRC grant EP/P021247/1.}
\orcid{0000-0002-2998-0933}
\affiliation{%
  \institution{Department of Computer Science, University of Durham}
  \city{Durham}
  \country{UK}
}

\renewcommand{\shortauthors}{Henry Austin, Maximilien Gadouleau, George B. Mertzios, and Amitabh Trehan}


\begin{abstract}
    Broadcast is a central problem in distributed computing. Recently, Hussak and Trehan [PODC'19/DC'23] proposed a stateless broadcasting protocol (Amnesiac Flooding), which was surprisingly proven to terminate in asymptotically optimal time (linear in the diameter of the network). However, it remains unclear: (i) Are there other stateless terminating broadcast algorithms with the desirable properties of Amnesiac Flooding, (ii) How robust is Amnesiac Flooding with respect to \emph{faults}?

    In this paper we make progress on both of these fronts. Under a reasonable restriction (obliviousness to message content) additional to the fault-free synchronous model, we prove that Amnesiac Flooding is the \emph{only} strictly stateless deterministic protocol that can achieve terminating broadcast. We achieve this by identifying four natural properties of a terminating broadcast protocol that Amnesiac Flooding uniquely satisfies. In contrast, we prove that even minor relaxations of \textit{any} of these four criteria allow the construction of other terminating broadcast protocols.
    
    On the other hand, we prove that Amnesiac Flooding can become non-terminating or non-broadcasting, even if we allow just one node to drop a single message on a single edge in a single round. As a tool for proving this, we focus on the set of all \textit{configurations} of transmissions between nodes in the network, and obtain a \textit{dichotomy} characterizing the configurations, starting from which, Amnesiac Flooding terminates. 
    Additionally, we characterise the structure of sets of Byzantine agents capable of forcing non-termination or non-broadcast of the protocol on arbitrary networks.
\end{abstract}




\maketitle

\newpage
\setcounter{page}{1}

\section{Introduction}
The dissemination of information to disparate participants is a fundamental problem in both the construction and theory of distributed systems. A common strategy for solving this problem is to ``broadcast'', i.e.~to transmit a piece of information initially held by one agent to all other agents in the system~\cite{Attiya-WelchBook, Lynchbook, peleg,tanenbaum2011computer, GerardTelDistributedAlgosBook}. In fact, broadcast is not merely a fundamental communication primitive in many models, but also underlies solutions to other fundamental problems such as leader election and wake-up. Given this essential role in the operation of distributed computer systems and the potential volume of broadcasts, an important consideration is simplifying the algorithms and minimizing the overhead required for each broadcast.

Within a synchronous setting, \textit{Amnesiac Flooding} as introduced by Hussak and Trehan in 2019~\cite{HussakT-AF-PODC19,HussakT-AF-STACS20} eliminates the need of the standard flooding algorithm to store historical messages. The algorithm terminates in asymptotically optimal $O(D)$ time (for $D$ the diameter of the network) and is stateless as agents are not required to hold any information between communication rounds. The algorithm in the fault-free synchronous message passing model is defined as follows:

\begin{definition}{\textbf{Amnesiac flooding algorithm.}}\label{def: AF} (adapted from~\cite{hussak2023termination})
Let $G = (V,E)$ be an undirected graph, with vertices $V$ and edges $E$ (representing a network where the vertices represent the nodes of the network and edges represent the connections between the nodes).  Computation proceeds in synchronous `rounds' where each round consists of nodes receiving messages sent from their neighbours. A receiving node then sends messages to some neighbours which arrive in the next round. No messages are lost in transit. The algorithm is defined by the following rules:
\begin{itemize}
\item[(i)]
All  nodes from a subset of sources or {\it initial nodes} $I \subseteq V$  send a message $M$ to all of their neighbours in round 1. 
\item[(ii)]
In subsequent rounds, every node that received $M$ from a neighbour in the previous round, sends $M$ to all, and only, those nodes from which it did not receive $M$. Flooding \emph{terminates} when $M$ is no longer sent to any node in the network.
\end{itemize}
\end{definition}

These rules imply several other desirable properties. Firstly, the algorithm only requires the ability to forward the messages, but does not read the content (or even the header information) of any message to make routing decisions. Secondly, the algorithm only makes use of local information and does not require knowledge of a unique identifier. Thirdly, once the broadcast has begun, the initial broadcaster may immediately forget that they started it.

However, extending Amnesiac Flooding and other stateless flooding algorithms (such as those proposed in~\cite{turau2020stateless,turau2021synchronous,bayramzadeh2021weak}) beyond synchronous fault-free scenarios is challenging. This is due to the fragility of these algorithms and their inability to build in complex fault-tolerance due to the absence of state and longer term memory. It has subsequently been shown that no stateless flooding protocol can terminate under moderate asynchrony, unless it is allowed to perpetually modify a super-constant number (i.e.~$\omega(1)$) of bits in each message~\cite{turau2020stateless}. Yet, given the fundamental role of broadcast in distributed computing, the resilience of these protocols is extremely important even on synchronous networks.

Outside of a partial robustness to crash failures, the fault sensitivity of Amnesiac Flooding under synchrony has not been explored in the literature. This omission is further compounded by the use of Amnesiac Flooding as an underlying subroutine for the construction of other broadcast protocols.
In particular, multiple attempts have been made to extend Amnesiac Flooding to new settings (for example routing multiple concurrent broadcasts \cite{bayramzadeh2021weak} or flooding networks without guaranteed edge availability \cite{turau2021synchronous}), while maintaining its desirable properties. 
However, none have been entirely successful, typically requiring some state-fulness.
It has not in fact been established that \textit{any other} protocol can retain all of Amnesiac Flooding's remarkable properties even in its original setting.
These gaps stem fundamentally from the currently limited knowledge of the dynamics of Amnesiac Flooding beyond the fact of its termination and its speed to do so.
In particular, all of the existing techniques (e.g. parity arguments such as in~\cite{hussak2023termination} and auxiliary graph constructions such as in~\cite{turau2021amnesiac}) used to obtain termination results for Amnesiac Flooding are unable to consider faulty executions of the protocol and fail to capture the underlying structures driving terminating behaviour.\\

We address these gaps through the application of novel analysis and by considering the structural properties of Amnesiac Flooding directly. 
By considering the sequence of message configurations, we are able to identify the structures underlying Amnesiac Flooding's termination and use these to reason about the algorithm in arbitrary configurations.
The resulting dichotomy gives a comprehensive and structured understanding of termination in Amnesiac Flooding.
For example, we apply this to investigate the sensitivity of Amnesiac Flooding with respect to several forms of fault and find it to be quite fragile.
Furthermore, we show that under reasonable assumptions on the properties of a synchronous network, any strictly stateless deterministic terminating broadcast algorithm oblivious to the content of messages, must produce the exact same sequence of message configurations as Amnesiac Flooding on any network from any initiator.
We therefore argue that Amnesiac Flooding is unique.
However, we show that if any of these restrictions are relaxed, even slightly, distinct terminating broadcast algorithms can be obtained.
As a result of this uniqueness and simplicity, we argue that Amnesiac Flooding represents a prototypical broadcast algorithm.
This leaves open the natural question: do there exist fundamental stateless algorithms underlying solutions to other canonical distributed network problems? Though memory can be essential or naturally useful in certain scenarios~\cite{MemoryLowerBoundsDISC,LimitsofInformationSpread,korman2022early,RandomisedBroadcastFirst,RandomisedBroadcastFinal,SocialNetworksSpreadRumors}, understanding what we can do with statelessness can help us push fundamental boundaries.

\subsection{Our Contributions}

In this work,
we investigate the existence of other protocols possessing the following four desirable properties of Amnesiac Flooding:
\begin{enumerate}
    \item \StrictStatelessness:
    Nodes maintain no information other than their port labellings between rounds. This includes whether or not they were in the initiator set.
    \item \Obliviousness: Routing decisions may not depend on the contents of received messages.
    \item \Determinism: All decisions made by a node must be deterministic.
    \item \Bandwidth: Each node may send at most one message per edge per round.
\end{enumerate}

Our main technical results regarding the existence of alternative protocols to Amnesiac Flooding are given in the next two theorems (reworded in \cref{sec: Uniqueness,Relaxing}).

\begin{theorem}[Uniqueness of Amnesiac Flooding]
    Any terminating broadcast algorithm possessing all of \StrictStatelessness, \Obliviousness, \Determinism\ and \Bandwidth{} behaves identically to Amnesiac Flooding on all graphs under all valid labellings.
\end{theorem}
Note that the last theorem allows, but does not require, that nodes have access to unique identifiers labelling themselves and their ports. 
However, we enforce the condition that these identifiers, should they exist, may be drawn adversarially from some super set of $[n+\kappa]$ where $n$ is the number of nodes on the networks and $\kappa=R(9,8)$ where $R(9,8)$ is a Ramsey Number.
It is important to stress here that this result holds even if the space of unique identifiers is only greater than $n$ by an additive constant.
In contrast to the last theorem, the next one does not assume that agents have access to unique identifiers.
\begin{theorem}[Existence of relaxed Algorithms]
    There exist terminating broadcast algorithms which behave distinctly from Amnesiac Flooding on infinitely many networks possessing any three of: \StrictStatelessness, \Obliviousness, \Determinism\ and \Bandwidth.
\end{theorem}

We derive four of these relaxed algorithms which all build upon Amnesiac Flooding:
\begin{enumerate}
    \item   \textsc{Neighbourhood-2 Flooding:} \StrictStatelessness\  is relaxed and agents are given knowledge of their neighbours' neighbours. This allows for distinct behaviour on networks of radius one as some agents are aware of the entire network topology.
    \item   \textsc{Random Flooding:} \Determinism\  is relaxed and agents are given access to a random coin. Agents randomly choose in each round whether to use Amnesiac Flooding or to forward messages to all of their neighbours.
    \item   \textsc{1-Bit Flooding (Message Dependent)}: \Obliviousness\  is relaxed and the source is allowed to include one bit of read-only control information in the messages. The source records in the message whether they are a leaf vertex, and if so agents perform Amnesiac Flooding upon receiving the message. Otherwise, agents perform a modified version where leaf vertices return messages.
    \item   \textsc{1-Bit Flooding (High Bandwidth):} \Bandwidth\ is relaxed and agents are allowed to send either one or two messages over an edge. This permits the same algorithm as the previous case by encoding the control information in the number of messages sent.
\end{enumerate}

We note that despite being a non-deterministic algorithm \textsc{Random Flooding} achieves broadcast with certainty and terminates almost surely in finite time.\\

We also perform a comprehensive investigation of the fault sensitivity of Amnesiac Flooding in a synchronous setting. 
Through the use of a method of invariants, we obtain much stronger characterizations of termination than were previously known, for both Amnesiac Flooding, and a subsequently proposed Stateless Flooding protocol~\cite{turau2020stateless}.  This allows us to provide precise characterizations of the behaviour of Amnesiac Flooding under the loss of single messages, uni-directional link failure, and time bounded Byzantine failures. 
The above invariants may be of independent interest, beyond fault sensitivity, as they provide strong intuition for how asynchrony interferes with the termination of both Amnesiac Flooding and the Stateless Flooding proposed in~\cite{turau2020stateless}.


The main technical result concerning fault sensitivity  is a \textit{dichotomy} characterizing the configurations, starting from which, Amnesiac Flooding terminates. 
As the rigorous statement of the result requires some additional notation and terminology, we will state it only informally here. 
We show in \cref{thm:balance} that, whether Amnesiac Flooding terminates when begun from a configuration or not, is determined exclusively by the parity of messages distributed around cycles and so-called faux-even cycles (FEC) (essentially pairs of disjoint odd cycles connected by a path, see section \ref{sec: faults} for a full definition) within the graph. 
This result implies the following three theorems which demonstrate the fragility of Amnesiac Flooding under three increasingly strong forms of fault.  We give our fault model explicitly in \cref{sec:faults}.
    \begin{theorem}[Single Message Failure]
    \label{thm:message drop}
        If single-source Amnesiac Flooding experiences a \emph{single message drop failure} for the message $(u,v)$ then it fails to 
       terminate
        if either:
        \begin{itemize}
            \item $uv$ is not a bridge
            \item $uv$ lies on a path between odd cycles
        \end{itemize}
        Moreover, it fails to broadcast if and only if this is the first message sent along $uv$, $uv$ is a bridge, and the side of the cut containing $u$ does not contain an odd cycle.
    \end{theorem}
We also consider the possibility of a link/edge failing in one direction (or if we consider an undirected edge as two directed links in opposite directions, only one of the directed links fails).
    \begin{theorem}[Uni-directional link failure]
    \label{thm: directed}
        For any graph $G=(V,E)$ and any initiator set $I\subsetneq V$ there exists an edge $e \in E$ such that a uni-directional link failure at $e$ will cause Amnesiac Flooding to either fail to broadcast or fail to terminate when initiated from $I$ on $G$. Furthermore, for any non-empty set of uni-directional link failures there exists $v \in V$ such that, when Amnesiac Flooding is initiated at $v$, it will either fail to broadcast or fail to terminate.
    \end{theorem}
We, finally, consider a set of Byzantine nodes, who know the original message, are free to collude among themselves and may decide to forward this message in any arbitrary pattern to their neighbours. However, for a discussion of termination to be meaningful, we require that the nodes have Byzantine behaviour for only a finite number of rounds.
    \begin{theorem}[Byzantine Failure]\label{thm: Byzantine}
        If Amnesiac Flooding on $G=(V,E)$ initiated from $I\subsetneq V$ experiences a \emph{weak Byzantine failure} at $J\subseteq V\setminus I$, then the adversary can force:
        \begin{itemize}
            \item Failure to broadcast if and only if $J$ contains a cut vertex set.
            \item Non-termination if and only if at least one member of $J$ lies on either a cycle or a path between odd-cycles.
        \end{itemize}
    \end{theorem}
Two natural corollaries of the Single Message Failure Theorem we wish to highlight here are: (i) on any network from any initial node there exists a single message, the dropping of which, will produce either non-termination or non-broadcast and (ii) dropping any message on a bipartite network will cause either non-termination or non-broadcast.
Similarly, the Byzantine Failure Theorem implies that any Byzantine set containing a non-leaf node on any network can force either non-termination or non-broadcast for Amnesiac Flooding with any initial set.\\

\anew{
\textbf{Organisation of the paper:} The initial part of the paper presents the required technical and motivational background, statements of our results and a technical outline of the more interesting proofs. The detailed proofs are presented in the appendix. Related work is presented in Section~\ref{sec: relatedwork}. Section~\ref{sec: Model} presents the model and notation required for the following technical sections. \emph{Uniqueness} of the algorithm is discussed in Section~\ref{sec: Uniqueness} with full proofs given in appendix~\ref{apx: Uniqueness}. We relax the conditions  individually to derive some more stateless protocols in Section~\ref{subsec: relaxation} with the detailed proofs in appendix~\ref{apx: relaxed}. We present our work on the \emph{fault sensitivity} of amnesiac flooding in Section~\ref{sec: faults} with the detailed proofs in appendix~\ref{apx: faults}. However, for technical reasons, we give the full proof of ~\cref{thm:balance} earlier in appendix~\ref{apx: balance}. We end with our conclusions and pointers to future work in Section~\ref{sec: conclusions}.
}

\section{Related Work}
\label{sec: relatedwork}
The literature surrounding both the broadcast problem and fault sensitivity is vast, and a summary of even their intersection is well beyond the scope of this work. Instead, we shall focus on work specifically concerning stateless, or nearly stateless broadcast.

The termination of Amnesiac Flooding and its derivatives has been the focus of several works. In combination they provide the following result that Amnesiac Flooding terminates under any sequence of multi-casts, i.e.
\begin{theorem}[Termination of Amnesiac Flooding (adapted from \cite{HT-AFCases-Arxiv20, hussak2023termination, turau2021amnesiac})]
    \label{Thm: termination}
    For $G=(V,E)$ a graph and $I_1,...,I_k \subseteq V$ a sequence of sets initiator nodes, Amnesiac Flooding on $G$ terminates when initiated from $I_1$ in round $1$, then $I_2$ in round $2$ and so forth.
\end{theorem}

Two independent proofs of the algorithm's termination have been presented, using either parity arguments over message return times~\cite{hussak2023termination} or axillary graph constructions~\cite{turau2021amnesiac}. The latter technique has further been used to establish tight diameter independent bounds on the termination of multicast using Amnesiac Flooding, complementing the eccentricity based bounds of~\cite{hussak2023termination}. The techniques we develop in this work, however, are more closely aligned with those of~\cite{HT-AFCases-Arxiv20}, as we exploit a similar notion to their "even flooding cycles" in our message path argument. In contrast to that work we focus on arbitrary configurations of messages, rather than just those resulting purely from a correct broadcast which allows us to obtain a much stronger characterisation. Combining our techniques with the dual "reverse" flood introduced by~\cite{HT-AFCases-Arxiv20}, we are able to show the complement of \cref{Thm: termination}, that only those configurations reachable from a sequence of multi-casts lead to termination.

There have, further, been multiple variants of Amnesiac Flooding introduced. It was observed by~\cite{turau2020stateless} in a result reminiscent of the BASIC protocol proposed by~\cite{gopal1999fast}, that sending a second wave of messages from a subset of the initial nodes could reduce the worst case $2D+1$ termination time to the optimal $D+1$ in all but a specific subset of bipartite graphs. We note that our fault sensitivity results extend naturally to this algorithm as well, as the same invariants apply to this setting. Beyond this, there have been several approaches to deal with the flooding of multiple messages simultaneously. In~\cite{HT-AFCases-Arxiv20}, the original authors of~\cite{HussakT-AF-PODC19} show that under certain conditions termination can be retained, even when conflicting floods occur. Since then, two partially stateless algorithms have been proposed, both making use of message buffers and a small amount of local memory~\cite{turau2021synchronous, bayramzadeh2021weak}. We will not be directly concerned with these approaches, however, as we assume a single concurrent broadcast throughout. However, the mechanism employed in~\cite{turau2021synchronous} should be highlighted as it rather cleverly exploits the underlying parity properties we identify as driving termination. Furthermore, as reduction to Amnesiac Flooding is used as a technique for proofs in many of these works, the comprehensive understanding of its termination we present here could prove a powerful tool for future work in these areas.

While the robustness of Amnesiac flooding and its variants have been previously studied, this has been focused on two forms of fault. The first is the disappearance and reappearance of nodes and links. The termination of Amnesiac Flooding is robust to disappearance and vulnerable to reappearance~\cite{hussak2023termination}. We will observe that  this is a necessary consequence of the invariants driving termination and their relation to cyclicity. In particular, the disappearance of nodes and links cannot form new cycles violating the invariant, whereas their reappearance can. A pseudo-stateless extension to Amnesiac Flooding has been proposed to circumvent this~\cite{turau2021synchronous}, implicitly exploiting the parity conditions of~\cite{hussak2023termination}. The second are faults that violate synchrony. Under a strong form of asynchrony, truly stateless and terminating broadcast is impossible~\cite{turau2020stateless}. However, the landscape under a weaker form of asynchrony (namely, the case of fixed delays on communication links) is more fine-grained. Although termination results have been obtained for cycles, as well as the case of single delayed edges in bipartite graphs~\cite{hussak2023termination}, there is no clear understanding of the impact of fixed channel delays. While we do not directly address this, we believe that techniques mirroring our invariant characterizations may prove fruitful in this area. To our knowledge, this work is the first to consider both the uniqueness of Amnesiac Flooding, as well as its fault sensitivity beyond node/link unavailability in a synchronous setting.

Beyond Amnesiac Flooding and its extensions, the role of memory in information dissemination is well studied in a variety of contexts. Frequently, stateful methods obtain faster termination time, such as in the phone-call model where the ability to remember one's communication partners and prevent re-communication dramatically improves termination time and message efficiency~\cite{RandomisedBroadcastFirst,RandomisedBroadcastFinal,SocialNetworksSpreadRumors}. 
Similarly, for bit dissemination in the passive communication model the addition of only $\log{\log{n}}$ bits of memory is sufficient to break the near linear time convergence lower bound of~\cite{LimitsofInformationSpread} and achieve polylogarithmic time~\cite{korman2022early}. 
Even more strongly, a recent work~\cite{MemoryLowerBoundsDISC} has shown that in the context of synchronous anonymous dynamic networks, stabilizing broadcast from an idle start is impossible with $O(1)$ memory and even with $o(\log{n})$ memory if termination detection is required.
Despite this, low memory and even stateless broadcasts remain desirable~\cite{gopal1999fast}. 
The possibility of solving other canonical distributed computing problems beyond broadcast, in a stateless manner, remains intriguing. In this direction, various low memory stateful models have been proposed to handle more complex distributed problems e.g. the compact local streaming (CLS) model in~\cite{CastanedaLT20-CLS-ICDCN20} with deterministic solutions (routing, self-healing fault-tolerance etc.) and randomised solutions for distributed colouring in a similar model~\cite{Flin0HKN23-LowMemColouring-SPAA23}. 

Several stateless broadcast schemes exist. To give an example, they are used in mobile ad hoc networks, which, due to the lower power and rapid movement of devices, see diminishing returns from maintaining information about the network~\cite{manfredi2011understanding}. 
However, given a lack of synchronisation as well as the wish to avoid so-called broadcast storms~\cite{ni1999broadcast}, these techniques typically rely on either some form of global knowledge (such as the direction or distance to the initiator) or the ability to sample network properties by eavesdropping on communications over time~\cite{manetsArentReallyStateless,ruiz2015survey}.
It should be noted however, that in contrast to many models, such as anonymous dynamic networks, radio networks and many manets, the typical framework for studying stateless flooding ("true statelessness" as defined by~\cite{turau2020stateless} restricting the model of~\cite{dolev2017stateless}) permits the knowledge and distinguishing of neighbours in both broadcasting and receiving.


\section{Model and Notation}
\label{sec: Model}
Throughout this work we consider only finite, connected graphs on at least two nodes.
We denote the set $\{1,...,x\}$ by $[x]$ and $R(r,s)$ the Ramsey number such that any graph of size $R(r,s)$ contains either a clique on $r$ vertices or an independent set on $s$ vertices. In this work, we make use of a generic synchronous message passing model with several additional assumptions based on the truly stateless model of \cite{turau2020stateless}. In particular, nodes cannot maintain any additional information between rounds (such as routing information, previous participation in the flood or even a clock value), cannot hold onto messages and can only forward, not modify the messages.\\
Unless stated otherwise we do not assume that nodes have access to unique identifiers, however they have locally labelled ports that are distinguishable and totally orderable for both receiving and sending messages. When we do work with identifiers these identifiers are assumed to be unique and drawn from $[|V|+\kappa]$ for $\kappa>0$ a constant. However, we assume that individual nodes have access to arbitrarily powerful computation on information they do have.\\
We are principally interested in the problem of \textit{broadcast}, although we will occasionally consider the related \emph{multicast} problem \anew{i.e. there are multiple initiators who may potentially even wake up in different rounds with the same message to be broadcast}. For a graph $G=(V,E)$ and an \textit{initiator} set $I\subseteq V$ we say that a node is informed if it has ever received a message from a previously informed node (where initiators are assumed to begin informed). An algorithm \textit{correctly} solves broadcast (resp. multicast) on $G$ if for all singleton (resp. non-empty) initiator sets there exists a finite number of rounds after which all nodes will be informed. Unless specified otherwise, we assume that initiator nodes remain aware of their membership for only a single round. We say that an algorithm \textit{terminates} on $G=(V,E)$ if, for all valid initiator sets, there exists a finite round after which no further messages are sent (i.e. the communication network quiesces).\\
\anew{Formally, for the message $M$, we denote a \textit{configuration} of Amnesiac Flooding as follows:
\begin{definition}
A configuration of Amnesiac Flooding on graph $G=(V,E)$ is a collection of \textit{messages/edges} 
$S\subseteq \{(u,v)| uv \in E\}$ where $(u,v) \in S$ implies that in the current round $u$ sent a message to $v$.
\end{definition}
Further, for $H$ a subgraph of $G$, we denote by $S_H$, $S$ restricted to $H$. We will refer to $u$ as the \emph{head} of the message and $v$ as its \emph{tail}.
Below, we define the operator $A_{I,G}:2^{V^2}\to 2^{V^2}$ to implement one round of Amnesiac Flooding on the given configuration:
\begin{definition} 
The operator $A_{I,G}:2^{V^2}\to 2^{V^2}$ is defined as follows: The set $I$ of nodes initiate the broadcast. The message $(u,v) \in A_{I,G}(S)$ if  $uv \in E$, $(v,u)\notin S$ and either $\exists w \in V: (w,u) \in S$ or $u \in I$.
\end{definition}

We will drop the subscript when $G$ is obvious from context and $I=\emptyset$. The multicast termination result of \cite{hussak2023termination} in this notation can be expressed as follows:
\begin{lemma}[Multicast Amnesiac Flooding]
    For any graph $G=(V,E)$, and any finite sequence $I_1,...,I_k\subseteq V$, there exists $m \in \mathbb{N}$ such that $$A^m_{\emptyset,G}(A_{I_k,G}(...A_{I_1,G}(\emptyset)))=\emptyset.$$
\end{lemma}}

\section{Uniqueness}
\label{sec: Uniqueness}
In this section, we investigate broadcast protocols similar to Amnesiac Flooding and establish four desirable properties that Amnesiac Flooding uniquely satisfies in combination. On the other hand, we show that this result is sharp and that by relaxing any of these conditions one can obtain similar terminating broadcast protocols.
\subsection{Uniqueness}
\label{subsec: Uniqueness}
Our first major result concerns the uniqueness of Amnesiac Flooding. Given the algorithms surprising properties, a natural question is whether other broadcast algorithms exist maintaining these properties. Specifically, does there exist a terminating protocol for broadcast which obeys all of the following for all graphs and valid port labellings:
    \begin{enumerate}
        \label{obs:properties}
        \item  \StrongTrueStatelessness:
        Nodes maintain no information other than their port labellings between rounds. This includes whether or not they were in the initiator set.
        \item \Obliviousness: Routing decisions may not depend on the contents of received messages.
        \item \Determinism: All decisions made by a node must be deterministic.
        \item \Bandwidth: Each node may send at most one message per edge per round.
    \end{enumerate}

The answer is strongly negative. In fact there is no other terminating broadcast protocol preserving these properties at all. We will actually prove the slightly stronger case, that this holds even if agents are provided with unique identifiers and are aware of the identifiers of their neighbours. Intuitively, the \StrictStatelessness{} condition forces any broadcast protocol to make its forwarding decisions based only on the messages it receives in a given round.
The combination of \Obliviousness{} and \Bandwidth{} forces any protocol meeting the conditions to view messages as atomic.
Finally, \Determinism{} forces the protocol to make identical decisions every time it receives the same set of messages.
Formally, any broadcast protocol meeting the four conditions must be expressible in the following form:
\begin{definition}
\label{def: functions}
    A \emph{protocol} $P=(b,f)$ is a pair of functions, an \emph{initial function} $b$ and a \emph{forwarding function} $f$, where $b:\mathbb{N}\times 2^{\mathbb{N}}\to 2^{\mathbb{N}}$ and $f:\mathbb{N}\times 2^{\mathbb{N}}\times 2^\mathbb{N}\to 2^\mathbb{N}$. 
    The protocol is implemented as follows. On the first round the initiator node $s$ with neighbourhood $N(s)$ sends messages to $b(s,N(s))$.
    On future rounds, each node $u$ sends messages to every node in $f(u,N(u), S)$ where $S$ is the set of nodes they received messages from in the previous round.
    Further, we require that $b(u,N(u)), f(u,N(u),S)\subseteq N(u)$ and that $f(u,N(u),\emptyset)=\emptyset$, enforcing that agents can only communicate over edges of the graph and can only forward messages they have actually received respectively.\\
    If there exists a graph $G=(V,E)$ with a unique node labelling and initiator node $s$ such that protocols $P$ and $Q$ send different messages in at least one round when implementing broadcast on $G$ initiated at $s$, then we say that $P$ and $Q$ are \emph{distinct}.
    Otherwise, we consider them the same protocol.
\end{definition}
In this setting, achieving broadcast is equivalent to every node receiving a message at least once and terminating in finite time corresponds to there existing a finite round after which no messages are sent.
For example, Amnesiac Flooding is defined by the following functions:
\begin{definition}[Amnesiac Flooding Redefinition]
    $P_{AF}=(b_{AF},f_{AF})$, where $b(u,S)=S$ and $f(u,S,T)=S\setminus T$ if $T\neq \emptyset$ and $\emptyset$ otherwise.
\end{definition}
For a protocol $P=(b,f)$, a set $S \subset \mathbb{N}$ and a degree $k\in \mathbb{N}$ we describe $P$ as AF up to degree $k$ on $S$ if there is no graph of maximum degree $k$ or less labelled with only members of $S$ which distinguishes $P$ from $P_{AF}$.
From here on we will assume that all unique labels are drawn from $[n+\kappa]$ where $\kappa$ is a sufficiently large constant and $n$ is the number of nodes in the graph. We obtain the following result, discussion and sketching the proof of which makes up the remainder of this section, provided $\kappa\geq R(9,8)$.
\begin{theorem}
\label{thm: AF is unique}
    Let $P=(b,f)$ be a correct and terminating broadcast protocol defined according to definition \ref{def: functions}, then $P$ is not distinct from Amnesiac Flooding.
\end{theorem}
\begin{proof}[Proof sketch for \cref{thm: AF is unique}]
The basic argument is to show that any correct and terminating broadcast protocol that meets the criteria is identical to Amnesiac Flooding.
Our core technique is to construct a set of network topologies such that any policy distinct from Amnesiac Flooding fails on at least one of them.\\
For any given protocol we derive a directed graph (separate from the network topology) describing its behaviour and demonstrate via a forbidden subgraph argument that any set of IDs of size $R(x,8)$ must contain a subset $T$ of size at least $x$ that do not respond to each other as leaf nodes, i.e. $P$ is $AF$ up to degree $1$ on $T$.
We take $\kappa=R(9,8)$ and show that there must then exist $U\subseteq T$ containing at least $6$ identifiers such that $P$ is AF up to degree $2$ on $U$.
By constructing a set of small sub-cubic graphs, we are able to extend this to degree $3$.\\
These form the base case of a pair of inductive arguments. First, we construct a progression of sub-cubic graphs which enforce that if $P$ is AF on $[m]$ up to degree $3$ it must be AF on $[m+1]$ up to degree $3$.
We then construct a family of graphs which have a single node of high-degree, while all other nodes have a maximum degree of $3$ and so must behave as though running Amnesiac Flooding.
These graphs permit a second inductive argument showing that this unique high degree node must also behave as if running Amnesiac Flooding.
In combination, these two constructions enforce that $P$ behaves like Amnesiac Flooding in all possible cases.
The full proof is given in appendix \ref{apx: Uniqueness}.
\end{proof}

\subsection{Relaxing the constraints}
\label{subsec: relaxation}
\label{Relaxing}
Despite the uniqueness established in the previous subsection, we are able to derive four relaxed algorithms distinct from Amnesiac Flooding each obeying only three of the four conditions.
\begin{theorem}
    \label{thm: Relaxation}
    There exist correct and terminating broadcast algorithms distinct from Amnesiac Flooding on infinitely many networks obeying any three of \StrongTrueStatelessness, \Obliviousness, \Determinism{} and \Bandwidth{}.
\end{theorem}
The algorithms we obtain all build upon Amnesiac Flooding, and we believe illuminate the role of each of the four conditions in the uniqueness result by showing what they prevent. The algorithms for \Obliviousness{} and \Bandwidth{} are presented together, as they differ only in how control information is encoded.
\begin{itemize}
    \item \StrongTrueStatelessness: Several relaxations of this already exist, such as Stateless Flooding (the initiator retains information for one round) or even classical non-Amnesiac Flooding (nodes are able to retain 1-bit for one round). We present \textsc{Neighbourhood-2 Flooding}. Nodes know the ID of their neighbours' neighbours. The protocol behaves distinctly on star graphs, as the hub can determine the entire network topology.\label{relax: stateless}
    \item \Obliviousness\ and \Bandwidth: \textsc{1-bit flooding}. Nodes are allowed to send a single bit of read-only control information (in the message header or encoded in the number of messages sent) communicating whether the initiator is a leaf vertex. If it is, nodes implement Amnesiac Flooding, otherwise they use a different mechanism called \textsc{Parrot Flooding} (leaves bounce the message back) which always terminates when begun from a non-leaf vertex. \label{relax: blind}
    \item \Determinism: \textsc{Random-Flooding}. Nodes have access to one bit of randomness per round. Each round every node randomly chooses to implement Amnesiac Flooding or to forward to all neighbours. \textsc{Random-Flooding} is correct with certainty and terminates almost surely in finite time. \label{relax: determinism}
\end{itemize}
Since each of these constitutes only a minor relaxation of the restrictions, we argue that the uniqueness of Amnesiac Flooding is in some sense sharp.
We present the protocols fully and demonstrate their correctness and termination for each of these cases independently in appendix \ref{apx: relaxed}.

\section{Fault Sensitivity}
    \label{sec: faults}
    With the uniqueness of Amnesiac Flooding established, a greater understanding of its properties is warranted.
    In this section, we study the configuration space of Amnesiac Flooding and obtain an exact characterisation of terminating configurations. We then apply this to investigate the algorithm's fault sensitivity.
\subsection{Obtaining a termination dichotomy}
    In order to consider the fault sensitivity of Amnesiac Flooding, we need to be able to determine its behaviour outside of correct broadcasts. 
    Unfortunately, neither of the existing termination proofs naturally extend to the case of arbitrary message configurations.
    Fortunately, we can derive an invariant property of message configurations when restricted to subgraphs that exactly captures non-termination, which we will call "balance". 
    \begin{theorem}
        \label{thm:balance}
        Let $S$ be a configuration on $G=(V,E)$ then there exists $k\geq 0$ such that $A^{k}_{G}(S)=\emptyset$ if and only if $S$ is balanced on $G$.
    \end{theorem}
    In fact we obtain that not only do balanced configurations terminate, they terminate quickly.
    \begin{corollary}
    \label{corr: balance is fast}
        Let $S$ be a balanced configuration on $G=(V,E)$ then there exists $k\leq 2|E|$ such that $A^{k}_G(S)=\emptyset$.
    \end{corollary}
    Intuitively, for the protocol not to terminate, we require that a message is passed around forever and since it is impossible for a message to be passed back from a leaf node the message must traverse either a cycle or system of interconnected cycles. 
    As we will demonstrate in the rest of the section we need only consider systems of at most two cycles. Specifically, we can introduce an invariant property determined by parity constraints on the number of messages travelling in each direction and their spacing around: odd-cycles, even cycles and what we will dub, faux-even cycles.
    \begin{definition}
        A \emph{faux-even cycle} (FEC) is a graph comprised of either two node disjoint odd cycles connected by a path or two odd cycles sharing only a single node. We denote by $FEC_{x,y,z}$ the $FEC$ with one cycle of length $2x+1$, one of length $2z+1$ and a path containing $y$ edges between them.
        We emphasize that if $y=0$ the two cycles share a common node and if $y=1$ the two cycles are connected by a single edge.
    \end{definition}
    FECs get their name from behaving like even cycles with respect to the operator $A$. In order to capture this we can perform a transformation to convert them into an equivalent even cycle.
    \begin{definition}
        Let $F=(V,E)$ be $FEC_{x,y,z}$, the \emph{even cycle representation} of $F$ denoted $F_2$ is the graph constructed by splitting the end points of the interconnecting path in two, and duplicating the path to produce an even cycle. Formally, if the two cycles are of the form $a_0...a_{2x}a_0$ and $c_0...c_{2z}c_0$, with $a_0$ and $c_0$ connected by the path $b_1...b_{y-1}$, we construct the following large even cycle from four paths: $$a_0...a_{2x}a_{-1}d_1...d_{y-1}c_{-1}c_{2z}...c_0b_{y-1}...b_1a_0.$$ Here $a_{-1}$ is a copy of $a_0$, $c_{-1}$ is a copy of $c_0$, and the path $d_1...d_{y-1}$ is a copy of the path $b_1...b_{y-1}$ (See figure \ref{fig: Even Cycle Representation}). Note that if $y=0$, $a_0=c_0$ and so we do not include any nodes from $b$ or $d$. Similarly, if $y=1$, $a_0$ and $c_0$ are connected by a single edge, as are $a_{-1}$ and $c_{-1}$.  
    \end{definition}
    There then exists a corresponding message configuration over the even cycle representation. Essentially (other than a few technical exceptions), this new configuration is the same as the old configuration but with two copies of each message on the path, one on each corresponding edge of the even cycle representation.
    Formally,
    \begin{definition}
        Let $F=(V,E)$ be $FEC_{x,y,z}$ and $S$ a configuration of Amnesiac Flooding on $F$, the \emph{even cycle representation} of $S$ on $F$ denoted $S_{2,F}$ is determined as follows. For each $m \in S$
    \begin{itemize}
        \item If $m=(a_{2x},a_0)$ (resp. $(a_0,a_{2x})$) we add $(a_{2x},a_{-1})$ (resp. $(a_{-1},a_{2x})$) to $S_{2,F}$.
        \item If $m=(b_{i},b_{j})$ for some $i,j \in \{1,...,y-1\}$, we add both $(b_{i},b_{j})$ and $(d_i,d_j)$ to $S_{2,F}$.
        \item If $m=(c_{2z},c_0)$ (resp. $(c_0,c_{2z})$) we add $(c_{2z},c_{-1})$ (resp. $(c_{-1},c_{2z})$) to $S_{2,F}$.
        \item If $m=(a_0, b_1)$ (resp. $(b_1,a_0)$) we add both $(a_0,b_1)$ and $(a_{-1},d_1)$ (resp. $(b_1,a_0)$ and $(d_1,a_{-1})$) to $S_{2,F}$.
        \item If $m=(c_0, b_{y-1})$ (resp. $(b_{y-1},c_0)$) we add both $(c_0,b_{y-1})$ and $(c_{-1},d_{y-1})$ (resp. $(b_{y-1},c_0)$ and $(d_{y-1},c_{-1})$) to $S_{2,F}$.
        \item If $m=(a_0,c_0)$ (resp. $(c_0,a_0)$) we add both $(a_0,c_0)$ and $(a_{-1},c_{-1})$ (resp. $(c_{0},a_{0})$ and $(c_{-1},a_{-1})$) to $S_{2F}$.
        \item Otherwise we add $m$ to $S_{2F}$.
    \end{itemize}
    \end{definition}
    \begin{figure}
        \centering
        \includegraphics[scale=0.25]{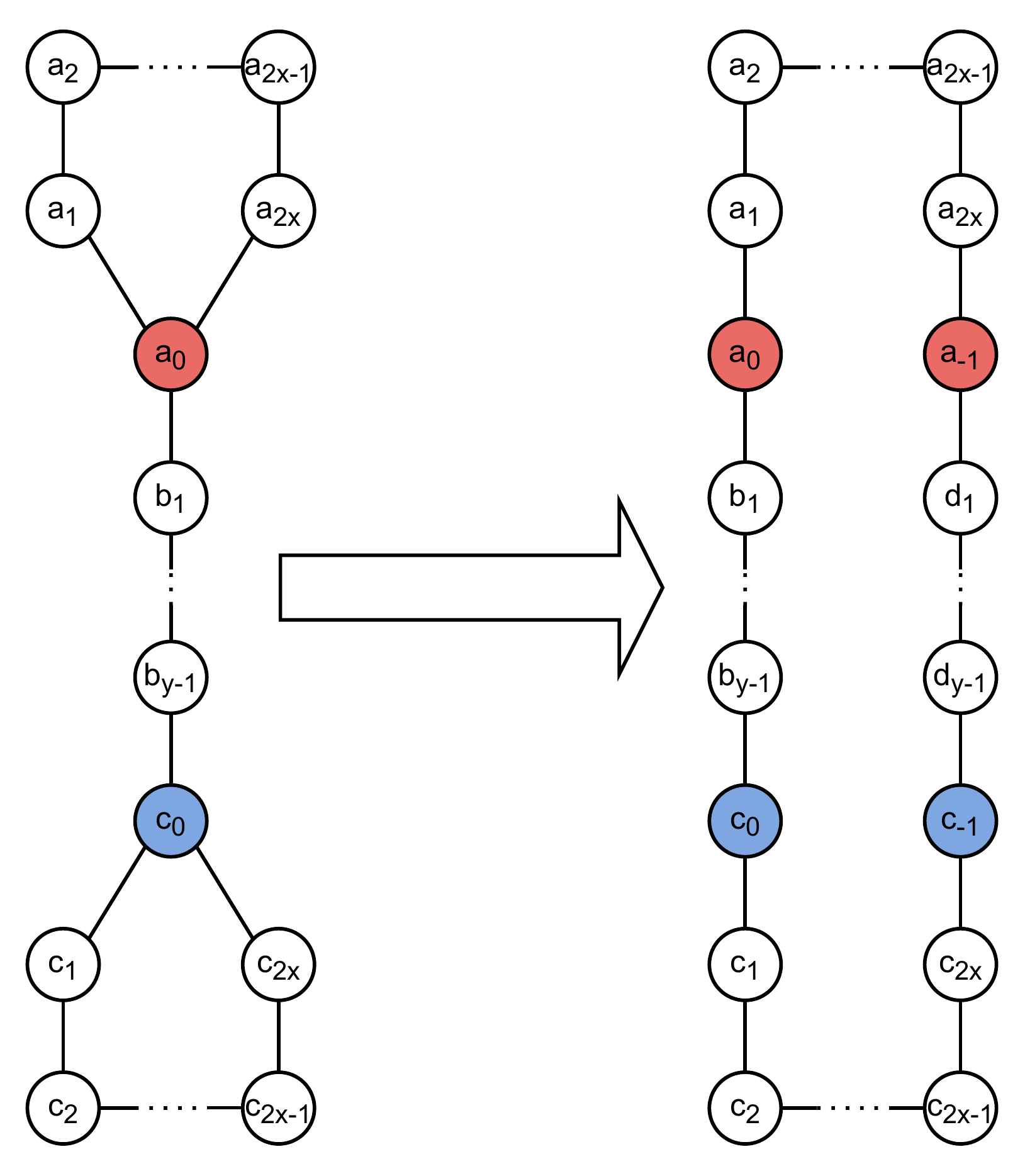}
        \caption{Left: An $FEC_{x,y,z}$. Right: The corresponding even cycle representation. Please note that this depiction only holds for $y\geq 2$. For $y=1$: $a_0$ and $c_0$ are connected directly by an edge in both sub figures (as are $a_{-1}$ and $c_{-1}$). For $y=0$: $a_0=c_0$ and $a_{-1}=c_{-1}$. }
        \label{fig: Even Cycle Representation}
    \end{figure}
    With this established we can now define the notion of balance.
    \begin{definition}
    \label{def: balance}
        A configuration $S$ is \emph{balanced} on $G=(V,E)$ if for all subgraphs $H$ of $G$ one of the following holds:
        \begin{itemize}
            \item $H$ is not a cycle or FEC.
            \item $H$ is an odd cycle and $S_H$ contains an equal number of messages travelling clockwise and anti-clockwise on $H$.
            \item $H$ is an even cycle and for any given message $m$ in $S_H$, there is an equal number of messages travelling clockwise and anti-clockwise on $H$ such that their heads are an even distance from $m$'s.
            \item $H$ is an $FEC$ and $S_{2,H}$ is balanced on $H_2$.
        \end{itemize}
    \end{definition}
    With these definitions established, we can present the intuition behind the proof of \cref{thm:balance}.
    \begin{proof}[Sketch of the proof of \cref{thm:balance}]
    We first establish that balance, and therefore imbalance, is conserved by Amnesiac flooding and, as the empty configuration is balanced, Amnesiac Flooding cannot terminate from any imbalanced configuration.
    For Amnesiac Flooding not to terminate it requires that some message travels around the communication graph and returns to the same edge, in the same direction. We show that if a configuration is balanced, the trajectory of any message can spend only a bounded number of consecutive steps on any given cycle or FEC. However, we can also show that any message's trajectory which crosses the same edge twice in the same direction, must have spent a large number of consecutive steps on some cycle or FEC, and therefore could not have begun in a balanced configuration.
    Thus, Amnesiac Flooding started from any balanced configuration must terminate.
    The full proof is given in appendix \ref{apx: balance}.
    \end{proof}
    \subsection{Applying the termination dichotomy}
    In this section, we apply the dichotomy to obtain a number of results.
    \label{sec:faults}
    \subsubsection{Extended Dichotomy}
    While \cref{thm:balance} provides a full dichotomy over the configuration space of Amnesiac Flooding and is much easier to reason about than previous results, the definition is somewhat unwieldy. In this subsection, we demonstrate the effectiveness of the dichotomy and unify it with the existing results~\cite{hussak2023termination, turau2021amnesiac, HT-AFCases-Arxiv20}. It has previously been observed that running Amnesiac Flooding backwards obtains another instance of multi-cast Amnesiac Flooding~\cite{HT-AFCases-Arxiv20}. Formally,
    \begin{definition}
        Let $G=(V,E)$ be a graph and $S$ be a configuration of messages on $G$. Then $\bar{S}=\{(v,u)|(u,v) \in S\}$.
    \end{definition}
    \begin{lemma}
    \label{lemma: going backwards}
        Let $G=(V,E)$ be a graph, $S$ a configuration of messages on $G$ and $T=\{u \in V| \forall v \in N(v): (v,u) \in S\}$ the set of sink vertices. Then $A_{T,G}\overline{A_{\emptyset,G}(\bar{S})}=S$
    \end{lemma}
    Which gives the following immediately via induction,
    \begin{lemma}
        \label{lem: time travel}
        Let $G=(V,E)$ be a graph and $S$ a configuration of messages on $G$. Then for any $k \in \mathbb{N}$ there exists a sequence $I_1,...,I_k \subseteq V$ such that $A_{I_1,G}...A_{I_k,G}\overline{A^k_{\emptyset,G}(\bar{S})}=S$.
    \end{lemma}
    Intuitively, this means that given any configuration $S$ of Amnesiac Flooding, we can run it backwards through time to some earlier configuration $S'$. Further we obtain a sequence of vertex sets $I_1,...,I_k$ that were sinks in the time-reversed process and therefore sources in the forwards process. We can therefore reconstruct $S$ beginning from $S'$ via some sequence of fresh multi-casts from $I_1,...,I_k$. We will use this fact to obtain all configurations from which Amnesiac Flooding terminates (i.e. balanced configurations) from the empty configuration. The following lemma is immediate from the definition of balance (definition \ref{def: balance}), as reversing the direction of all messages in a configuration does not affect its balance.
    \begin{lemma}
        \label{lem: balance is invariant in time}
        $\bar{S}$ is balanced on $G$ if and only if $S$ is balanced on $G$
    \end{lemma}
    Putting it all together, we obtain the following extension of the dichotomy result, as well as the complement to \cref{Thm: termination}.
    \begin{theorem}
        Let $G=(V,E)$ be a graph and $S$ a configuration of $G$, the following are all equivalent:
        \begin{enumerate}
            \item $\exists k \in \mathbb{N}: A^k_{\emptyset,G}(S)=\emptyset$
            \item $\exists k \in \mathbb{N},  I_1,...,I_k\subseteq V: A_{I_k,G}...A_{I_1,G}(\emptyset)=S$
            \item $S$ is balanced on $G$
        \end{enumerate}
    \end{theorem}
    \begin{proof}
        The equivalence of (1) and (3) follow immediately from \cref{thm:balance}. Further, we have that (2) implies (1) from \cref{Thm: termination}. Now assume $S$ is balanced, then by lemma \ref{lem: balance is invariant in time}, so is $\bar{S}$. Thus by \cref{thm:balance} there exists a finite $k$ such that after $k$ rounds Amnesiac flooding started from $\bar{S}$ must terminate, i.e. $A^k_{\emptyset,G}(\bar{S})=\emptyset$ . Therefore, by lemma \ref{lem: time travel} we have a sequence $I_1,...,I_k \subseteq V$ such that $S=A_{I_1,G}...A_{I_k,G}\overline{A^k_{\emptyset,G}(\bar{S})}=A_{I_1,G}...A_{I_k,G}(\emptyset)$. Thus, we have (3) implies (2) and the result follows.
    \end{proof}
    \subsubsection{Fault Sensitivity}
    In this work we consider three key forms of fault of increasing severity: message dropping, uni-directional link failure and weak-Byzantine failures. Intuitively, these correspond to a set of messages failing to send in a specific round, a link failing in one direction creating a directed edge and a set of nodes becoming transiently controlled by an adversary.

    More precisely, let $\mathbf{S}=(S_{i})_{i \in \mathbb{N}}$ be the sequence of actual message configurations on our network. 
    We say that $\mathbf{S}$ is \emph{fault free} for $G=(V,E)$ \textbf{if} $S_{i+1}=A_{G}(S_{i})$ for all $i \in \mathbb{N}$.
    Otherwise, we say it experienced a \emph{fault}. We say $\mathbf{S}$ has suffered from,
     
    \begin{itemize}
        \item Message dropping, if there exists $T\subseteq V^2$ and $k\geq 1$ such that $S_{k+1}=A(S_k)\setminus T$ and for all $i\neq k$, $S_{i+1}=A(S_i)$. This corresponds to all messages in $T$ being dropped on round $k$.
        \item Uni-directional link failure, if there exists $X\subseteq V^2$ such that for all $i\geq 1$, $S_{i+1}=A(S_i)\setminus X$. This corresponds to all oriented links in $X$ failing.
        \item Weak-Byzantine, if there exists $Y \subseteq V$ such that for some $k$ at least twice the diameter, for all $i<k$, $S_{i+1}\setminus\{(u,v)|u \in Y\}=A(S_i)\setminus\{(u,v)|u \in Y\}$. This corresponds to a possible failure where an adversary determines the forwarding decisions of the nodes in $Y$ until round $k$.
    \end{itemize}
    Note that we refer to the Byzantine failures as weak, since they are transient and only interfere with the forwarding of the message, not its content. It is obvious to see that in a stateless setting there is no way to deal with a Byzantine fault that changes the message as there is no method to verify which message is authentic. Intuitively, in our setting, Weak-Byzantine agents may choose to send messages to an arbitrary set of neighbours in each round and they are all controlled by a single coordinated adversary. We say that a Weak-Byzantine adversary with control of a given set of nodes can force some behaviour if there exists any weak byzantine failure on that set of nodes producing the forced behaviour. We can now express our fault sensitivity results, the proofs of which we defer to appendix \ref{apx: faults}, and begin with an extreme case of single message dropping.
        \begin{theorem}[Single Message Failure]
    \label{thm:message drop}
        If single-source Amnesiac Flooding experiences a \emph{single message drop failure} for the message $(u,v)$ then it fails to 
       terminate
        if either:
        \begin{itemize}
            \item $uv$ is not a bridge
            \item $uv$ lies on a path between odd cycles
        \end{itemize}
        Moreover, it fails to broadcast if and only if this is the first message sent along $uv$, $uv$ is a bridge, and the side of the cut containing $u$ does not contain an odd cycle.
    \end{theorem}
    Thus, Amnesiac Flooding is extremely fault-sensitive with respect to message dropping.\\
    Secondly, considering uni-directional link failures we obtain the following.
        \begin{theorem}[Uni-directional link failure]
    \label{thm: directed}
        For any graph $G=(V,E)$ and any initiator set $I\subsetneq V$ there exists an edge $e \in E$ such that a uni-directional link failure at $e$ will cause Amnesiac Flooding to either fail to broadcast or fail to terminate when initiated from $I$ on $G$. Furthermore, for any non-empty set of uni-directional link failures there exists $v \in V$ such that, when Amnesiac Flooding is initiated at $v$, it will either fail to broadcast or fail to terminate.
    \end{theorem}
     Finally for the weak-Byzantine case.
        \begin{theorem}[Byzantine Failure]\label{thm: Byzantine}
        If Amnesiac Flooding on $G=(V,E)$ initiated from $I\subsetneq V$ experiences a \emph{weak Byzantine failure} at $J\subseteq V\setminus I$, then the adversary can force:
        \begin{itemize}
            \item Failure to broadcast if and only if $J$ contains a cut vertex set.
            \item Non-termination if and only if at least one member of $J$ lies on either a cycle or a path between odd-cycles.
        \end{itemize}
    \end{theorem}

\section{Conclusions and Future Work}
\label{sec: conclusions}
In this paper, we prove a \emph{uniqueness} result: Under standard synchronous message passing assumptions, any strictly stateless deterministic algorithm oblivious to the message content which solves terminating broadcast is indistinguishable from \emph{Amnesiac Flooding}. 
We therefore argue due to both its uniqueness and simplicity, that Amnesiac Flooding is a fundamental or prototypical broadcast algorithm.
We formalise the four properties required for this uniqueness to hold, and show that by relaxing each individually one can obtain other correct and terminating broadcast algorithms, of which we present several.
These present the following natural questions: To what extent does Amnesiac Flooding represent a ``minimal'' broadcast algorithm?
Are there identifiable families of algorithms solving terminating broadcast with a subset of these restrictions? 
Are any of these independent of (i.e. not derivatives of) Amnesiac Flooding? \\

We also obtain an understanding of the structural properties of Amnesiac Flooding. In particular, we study its sensitivity to single message drops, uni-directional link failures, and weak byzantine collusion, showing it can easily become non-terminating or non-broadcasting under such conditions. 
This is perhaps surprising, as statelessness is frequently associated with fault tolerance, such as in the self stabilizing setting.
A reasonable interpretation of \cref{thm:balance}, however, is that Amnesiac Flooding, while locally stateless, depends heavily on a distributed ``meta-state'' contained in the configuration of sent messages.
This suggests it is unlikely that any minor modification of Amnesiac Flooding will resolve its fragility without depending on an entirely different mechanism for termination.
In support of this, we note that of the four alternatives presented in the proof of the Existence of Relaxed Protocols Theorem, only Random-Flooding is meaningfully more robust (and will in fact terminate from any configuration in finite time almost surely).
Nevertheless, we contend that further exploration of stateless algorithms such as Amnesiac Flooding, their properties and related models are important for both theory and practice of distributed networks.

\appendix
\newpage
\begin{center}
    \huge{\textsc{Appendix}}\\
\end{center}
\medskip
The observant reader may notice that despite its statement being presented later in the body of the text, we prove \cref{thm:balance} first in the appendix.
We choose this ordering, as the result is used in one specific case in the proof of \cref{thm: Relaxation}.
\section{Proof of the Dichotomy (Theorem \ref{thm:balance})}
    \label{apx: balance}
    In this appendix we present all remaining definitions necessary for and prove theorem \cref{thm:balance}. First, we give the formal definition of the even cycle representation of a message configuration.
    
    With that established, we begin by showing that (im)balance is preserved by the operation of Amnesiac Flooding,
    \begin{lemma}
        For any set of initiators $I\subseteq V$ and configuration $S$ on $G=(V,E)$, $A_{I,G}(S)$ is balanced if and only if $S$ is balanced.
    \end{lemma}
    \begin{proof}
        We will show this for each subgraph type in the tetrachotomy of definition \ref{def: balance}. The first case is trivial. 
        
        Then, let $H$ be an odd cycle. If the number of messages on $H$ in $A_G(S)$ increases, this can only be due to a node $u$ on $H$ receiving a message from $v$ outside of $H$. In this case, if $u$ already received a message on $H$ it will have no effect, otherwise it will introduce two new messages, one clockwise and one anti-clockwise. If the number of messages decreases this means that two messages collided on $H$ removing both, this removes one clockwise and one anti-clockwise. Therefore, the balance on $H$ does not change.

        Now let $H$ be an even cycle. The parity of the distance between messages is preserved on an even cycle when $A$ is applied. Therefore, changing the balance would require the introduction or removal of messages. The argument is then the same as in the odd cycle case, with the caveat that any pair of messages added (removed) by the external activation of (collision at) a node have the same parity of distance between their heads and any other message. Therefore, the balance on $H$ is preserved.

        Finally, let $H$ be a faux-even cycle. If no nodes of $H$ receive a message from outside of $H$, then the even cycle configuration $A_H(S_H)_{2,H}$ is the same as $A_{H_2}(S_{2,H})$, unless there is a message from a cycle onto a shared node with the path and no other messages arriving at that node. This case is equivalent to an external activation of the node on the other node in $H_2$ associated with the shared path node. If new messages are added to $H$ from outside, these will produce either $2$ or $4$ messages on $H_2$ in the even cycle configuration, each pair in opposite directions and the same parity of distance apart. Therefore, the balance on $H_2$ is preserved.
    \end{proof}
    This gives the following immediate corollary, as the empty configuration is trivially balanced.
    \begin{corollary}
    \label{corr: Forward}
        If $S$ is imbalanced on $G=(V,E)$, then for all $k>0$, $A^k_G(S)\neq\emptyset$.
    \end{corollary}
    This gives us the forward direction of \cref{thm:balance}. For the other direction, we need the notion of message paths and their recurrence.
    \begin{definition}
    \label{def: message paths}
    A message $m=(v_0,v_1)$ in configuration $S\subset V^2$ has a \emph{message path} $v_0v_1$.
    We define the rest of its paths recursively, i.e. $m=(v_0,v_1)$ has a message path $v_0...v_{k+1}$ from $S$ on $G$ if:
    \begin{itemize}
    \item 
        \item $v_0v_1...v_{k}$ is a message path of $m$ in $S$
        \item The message $(v_k,v_{k+1})$ exists in $A_G^k(S)$
    \end{itemize}
     We say that a message $m=(v_0,v_1)$ is \emph{recurrent} on $G=(V,E)$ from $S$ if $m$ has a message path of the form $v_0v_1...v_0v_1$ on $G$ from $S$.
    \end{definition}
    We obtain the following property relating message paths and termination immediately.
    \begin{lemma}
    \label{lemma: recurrence}
        Let $S$ be a non-empty configuration on $G=(V,E)$ such that $A_{G}^k(S)=S$, then $S$ contains a recurrent message on $G$.
    \end{lemma}
    \begin{proof}
        If $m$ is a message in $A_{G}^k(S)$ there must exist a message $m'$ in $S$, such that there is a message path from $m'$ to $m$. By the same argument there must exist another message $m''$ such that there is a message path from $m''$ to $m'$ and therefore to $m$. However, this cannot extend back infinitely with unique messages as there is a finite number of messages in $S$. Thus, there must exist some repeated message in the sequence, which in turn is by definition recurrent.
    \end{proof}
    The crucial observation is that the number of consecutive steps that can be spent on certain subgraphs by a message path is bounded if beginning from a balanced configuration. We will need four such results.
    \begin{lemma}
    Let $G=(V,E)$ with $H \subseteq G$ an odd cycle of length $2k+1$. If $m$ is a message of $S$ on $H$ with a message path of length $2k+2$ restricted only to $H$, then $S$ is imbalanced on $G$.
        \label{lemma: odd cycles}
\end{lemma}
\begin{proof}
    For the sake of contradiction assume that $S$ is balanced on $G$ and $m\in S$ has a message path from $S$ of length $2k+2$ restricted only to $H$. 
    Without loss of generality, assume that $m$ is travelling clockwise.
    On a cycle a message can collide with exactly one other message, upon doing which all of its message paths must terminate. 
    If there is a message $m'$ on $H$ travelling anti-clockwise in $S$, then either it collides with $m$ within $2k$ steps or it collides with another message travelling clockwise first. 
    Since $S$ is balanced on $G$ there must be an equal number of messages travelling clockwise and anti-clockwise on $H$, and a collision removes one of each.
    Therefore, unless new messages are added to $H$ in subsequent steps, $m$ must collide with a message $m'$ travelling anti-clockwise before it takes $2k+1$ steps.
    Thus, new messages must be added to $H$ that collide with the message $m$ would otherwise collide with. 
    The introduction of a new message to $H$ must produce one message travelling clockwise and one travelling anti-clockwise, $m_c$ and $m_a$ respectively. 
    Thus, $m_c$ must collide with $m'$ before $m$ does and so must be introduced between $m$ and $m'$. 
    However, this means that $m$ will collide with $m_a$ before it would have reached $m'$ giving our contradiction.
\end{proof}
\begin{lemma}
    Let $G=(V,E)$ with $H \subseteq G$ an even cycle of length $2k$. If $m$ is a message of $S$ on $H$ with a message path of length $k+1$ restricted only to $H$, then $S$ is imbalanced on $G$.
    \label{lemma: even cycles}
\end{lemma}
\begin{proof}
    The argument is essentially the same as for the odd case with the caveat that in messages with an even distance between their heads travelling in opposite directions collide in at most $k$ steps on an even cycle. 
\end{proof}
\begin{lemma}
    Let $G=(v,E)$ with $H\subseteq G$ be a $FEC_{x,y,z}$. If $m$ is a message of $S$ on $H$ with a message path of length $x+y+z+2$ restricted only to $H$ but not only to one of its cycles, then $S$ is imbalanced on $G$.
    \label{lemma: FEC}
\end{lemma}
\begin{proof}
    For the sake of contradiction assume that $S$ is balanced on $G$ and such a message path exists. Then consider, a path of length $x+y+z+2$. Since $A_{H}(S)_{2,H}$ can only have additional messages relative to $A_{H_2}(S)$, a message path must deviate from $H_2$ in order for the Lemma \ref{lemma: even cycles} not to apply. Thus, such a message path must either remain on a cycle upon crossing the path point, or cross from one cycle onto the path and then to the other cycle in the wrong direction. In the former case this violates Lemma \ref{lemma: odd cycles}. In the latter case we can simply take the other even cycle representation and the result follows from Lemma \ref{lemma: even cycles}.
\end{proof}
\begin{lemma}
    Let $G=(V,E)$ with $H\subseteq G$ an even cycle of length $2k$ and let $m \in S$ be a message on $H$ where $S$ is balanced on $G$. If $m$ has a message path restricted to $H$ on $S$ of length $k$, then there exists $m'$ in $S$ on $H$ such that $m$ and $m'$ have the same start point and $m'$ also has a message path of length $k$ restricted to $H$.
    \label{lemma: EvenSub}
\end{lemma}
\begin{proof}
    Since $S$ is balanced on $G$ there must exist a message colliding with $m$ on $H$ exactly $k$ steps after $S$. Therefore, either this message exists in $S$ or is added to $H$ after $S$. If it is added after $S$, then there must already have existed a message travelling counter to $m$ in $S$ which would have collided with $m$ sooner. Further, any new message added to block this one would collide with $m$ even sooner. Iterating this argument, the statement holds.
\end{proof}
This gives us the wedge we will need to obtain our main result, which we now reframe according to Lemma \ref{lemma: recurrence}.
\begin{theorem}
    Let $S$ be a configuration on $G$, $S$ is imbalanced if and only if it contains a recurrent message.
\end{theorem}
The forward direction is an immediate consequence of Corollary \ref{corr: Forward}. For the reverse we take $G=(V,E)$ to be a communication graph and $S \subseteq V^2$ to be a balanced configuration of messages. For contradiction we assume that $S$ contains a recurrent message $m$ which has a message path $W$ performing exactly one excursion and return to $m$. We will view $W$ as both a walk on $G$ and a word. Further, $W$ must obey the Conditions of the lemmas defined previously re-framed here as:
\begin{enumerate}
    \item $W$ cannot return to the node it just came from, i.e. no sequence of the form $uvu$. \label{cond: No leaves}
    \item $W$ cannot take $2x+2$ consecutive steps around an odd cycle of length $2x+1$ (Lemma \ref{lemma: odd cycles}). \label{cond: odd cycle}
    \item $W$ cannot take $x+1$ consecutive steps around an even cycle of length $2x$ (Lemma \ref{lemma: even cycles}). \label{cond: even cycle}
    \item $W$ cannot take more than $x+y+z+2$ steps on an $FEC$ unless it sticks purely to one cycle (Lemma \ref{lemma: FEC}) .\label{cond: FEC}
    \item If $W$ takes $x$ steps around an even cycle of length $2x$, then there exists $W'$ which takes the opposite path of equal length and is otherwise identical (Lemma \ref{lemma: EvenSub}). \label{cond: even replacement}
\end{enumerate}
The fifth of these has the following useful interpretation, if the existence of $W'$ implies imbalance then the existence of $W$ implies imbalance. Therefore, we will use this rule as a substitution allowing us to "modify" $W$ to take the alternate path. We can now begin in earnest:
\begin{lemma}[Repetition Lemma]
    If $W$ contains a factor $u...u$ then that factor contains an odd cycle as a subfactor.
\end{lemma}
\begin{proof}
    By Condition \ref{cond: No leaves}, we cannot return immediately to the same node. If the whole factor is not a cycle we must return to a node before $u$ and so we instead consider the innermost cycle. This cycle must be consecutive and traversed fully, thus for $S$ to be balanced this cycle must be odd by Condition \ref{cond: even cycle}.
\end{proof}
\begin{lemma}[Odd Cycle Lemma]
\label{lemma: Odd Cycle Lemma}
    There can only be one factor of $W$ that forms an odd cycle.
\end{lemma}
\begin{proof}
    Assume there are at least two for contradiction. We take the first two cycles of odd length in $W$, they are either sequential ($u...u...v...v$), interleaved ($u...v...u...v$) or repetitive ($u...u...u$). In the sequential case this is a fully traversed $FEC$ and so in violation of Condition $4$. In the interleaved case, we have three subsequences $w_1,w_2,w_3$ of length $a,b,c$ respectively such that $a+b$ is odd and $b+c$ is odd (see Figure \ref{fig: Odd Cycle Lemma Diagram}). Thus, $a+c$ is even. Since $w_1$ and $w_3$ induce an even cycle, either $a=c$ or $W$ is in violation of Condition $3$. Therefore, we have that $W$ travels exactly half way around an even cycle along $w_3$. Thus there must exist a message path $W'$ that contains the factor $w_1w_2w_1$ which is more than $a+b$ steps around an odd cycle of length $a+b$ and so is in violation of Condition $2$. In the repetitive case, this is multiple traversals of an odd cycle in full and so also in violation of Condition $2$. Thus, there exists at most one odd cycle in $W$.
\end{proof}
\begin{figure}
    \centering
    \includegraphics[width=0.6\linewidth]{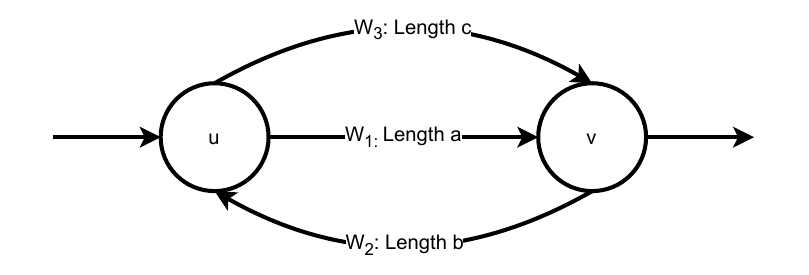}
    \caption{An illustration of the walk described in the interleaved case in the proof of \ref{lemma: Odd Cycle Lemma}}
    \label{fig: Odd Cycle Lemma Diagram}
\end{figure}
This immediately gives the following results:
\begin{lemma}
   Every node appears at most twice in $W$. \label{lemma: Three is a magic number}
\end{lemma}
\begin{lemma}
    There exists at most one node that is both a member of a consecutive odd cycle and appears twice. Furthermore, such a node must be the start and end of the cycle.
\end{lemma}
Therefore, there must exist a cycle $C$ containing $m$ such that $W$ fully traverses $C$ before returning to $m$, with possible excursions.
To be clear, there exists a sequence of pairs $(w_0,w_1),(w_1,w_2)...(w_k,w_0),(w_0,w_1)$ such that $C=w_0w_1..w_kw_0$ is a cycle of $G$, each pair appears in $W$ in consecutive order (i.e. $W=w_0w_1...w_1w_2...w_kw_0..w_0w_1$) and $m=(w_0,w_1)$. 
With this established we are now ready to prove \cref{thm:balance}.
We demonstrate that no matter how we construct $W$ it must take too many steps around some structure and therefore cannot exist.
\begin{proof}[Proof of \cref{thm:balance}]
    \begin{claim} 
        $C$ is an odd cycle
    \end{claim}
    \begin{claimproof}
    If $C$ is of even length (say $2k$), then there must be an excursion from $C$ otherwise Condition $3$ would be violated. However, there can be at most one such excursion as it must contain a consecutive odd cycle and so the two subsequences on either side must be factors of $W$. Therefore, since some subsequence of $W$ traverses $C$ fully with one additional step, one of the two factors must take $k+1$ steps around $C$. This violates Condition $3$ and so $C$ must be an odd cycle.
    \end{claimproof}
    \begin{claim} $W$ consists of two odd cycles $C$ and $\hat{C}$ connected by a path. \end{claim}
    \begin{claimproof}

    If $C$ is an odd cycle there must be an excursion from it or Condition $2$ would be violated. This excursion must contain exactly one consecutive odd cycle which we denote by $\hat{C}$. If $\hat{C}$ does not share its starting node with $C$, $W$ either forms a path between $C$ and $\hat{C}$ or some chain of cycles. We can use Condition $5$ to eliminate all even cycles of this chain, then if there are any odd cycles in the chain they correspond to a fully traversed $FEC$ when paired with $\hat{C}$ and so violate Condition $4$. Thus, we have a $W$ that takes a simple path from $C$ to $\hat{C}$ and back, although possibly intersecting $C$ along the way.
    \end{claimproof}
    \begin{claim}{Neither $C$ nor $\hat{C}$ intersect with the path between them. }\end{claim}
    \begin{claimproof}
    Assume that the path to $\hat{C}$ does intersect $C$, since we are taking the same path in both directions any node shared between the path and $C$ appear in $W$ three times. 
    This violates Lemma \ref{lemma: Three is a magic number} and so $C$ must be disjoint from the path to $\hat{C}$.
    Similarly $\hat{C}$ must be disjoint from the path otherwise it would violate the same lemma.
    If $C$ is disjoint from $\hat{C}$ the pair would form a fully traversed $FEC$, thereby violating Condition $4$. Thus, $C$ and $\hat{C}$ must intersect.
    \end{claimproof}
    \begin{claim}{Claim: $C$ and $\hat{C}$ do not intersect. }\end{claim}
    \begin{claimproof}
    Let $W=uv..wx_1,...,x_kw..uv$ where $C=u...w...uv$ and the excursion to $\hat{C}$ is given by $wx_1...x_kw$.
    Now assume that $\hat{C}$ contains a node from $C$ which occurs strictly before $w$ in $W$. This node is on a consecutive odd cycle and appears twice.
    There must exist a latest such node in the ordering of $C$, we denote it $y$.
    Since $y$ is on a consecutive odd cycle and appears twice in $W$ it must be the start and end point of $\hat{C}$. 
    However, then either $y=w$ which is impossible or $y$ appears three times, violating Lemma \ref{lemma: Three is a magic number}. 
    The same argument holds taking the earliest node shared by $C$ and $\hat{C}$ strictly after $w$. Thus, the only node that can be shared by $C$ and $\hat{C}$ is $W$, implying that $W$ forms two odd cycles sharing a single node.
    However, this is an $FEC$ which is fully traversed and so violates Condition $4$.\\
    Thus, $W$ cannot exist and so there can be no recurrent message in $S$.
    \end{claimproof}
\end{proof}
We can now immediately obtain a proof of corollary \ref{corr: balance is fast}.
\begin{proof}
    If there does not exist $k\leq 2|E|$ such that $A^k(S)=\emptyset$ then there must be a message path of length $2|E|+1$ for some message in $S$.
    By the pigeon hole principle this message path must visit the same edge twice travelling in the same direction.
    Thus, this requires a recurrent message.
    Therefore, $S$ cannot be balanced.
\end{proof}
\section{Proof of Uniqueness (Theorem \ref{thm: AF is unique})}
\label{apx: Uniqueness}
In this appendix we prove \cref{thm: AF is unique}.
We begin with the following immediate observations:
\begin{observation}
    \;
    \label{obs: basic uniqueness properties}
    \begin{enumerate}
        \item For all $\emptyset\neq S \subset \mathbb{N}$ and $u \in \mathbb{N}$, $b(u,S)\neq \emptyset$. Otherwise no new nodes will be informed after the first round, violating correctness.
        \item For $u, v, w \in \mathbb{N}$, $f(v,\{u,w\},\{u\})\in \{\{w\},\{u,w\}\}$. Otherwise if $v$ is a bridge on the graph the message can never reach the other side.
        \item For $u, v \in [\kappa]$, if $f(u,\{v\},\{v\})=v$ then $f(v,\{u\},\{u\})=\emptyset$. Otherwise $P$ will not terminate on the two node path labelled with $u$ and $v$.
        \item For $u,v,w \in [\kappa]$ if $f(u,\{v\},\{v\})=\{v\}$ and $f(w,\{v\},\{v\})=\{v\}$ then either $f(v,\{u,w\},\{u\})=\{u,w\}$, $f(v,\{u,w\},\{w\})=\{u,w\}$ or both. Additionally, $f(v,\{u,w\},\{u,w\})=\emptyset$. Otherwise the protocol would not terminate on the three node path with labels $uvw$. \label{obs:p3 behaviour}
    \end{enumerate}
\end{observation}
Our strategy will be to show the existence of a small set of labels which when restricted to $P$ behaves identically to Amnesiac Flooding for subcubic graphs. From there we can inductively argue that this holds for all labels up to degree $3$ and then in turn that this holds for all graphs. We begin with the following lemma, where for $u,v \in [\kappa]$, we denote $f(u,\{v\},\{v\})=\{v\}$ by $u>_*v$ for notational compactness, however please note that this relation is not transitive. 
\begin{lemma}
\label{lemma: ramsey}
    For any $c'$ if $\kappa$ is sufficiently large $\exists S \subseteq [\kappa]$ such that $\forall u,v \in S$ $f(u,\{v\},\{v\})=\emptyset$ and $|S|=c'$.
\end{lemma}
\begin{proof}
    We begin by considering $u,v,w,x,y,z \in [\kappa]$ such that $u>_*v<_*w$ and $x>_*y<_*z$. By  Observation \ref{obs: basic uniqueness properties} (4), wlog. we can assume that $f(v,\{u,w\},\{u\})=\{u,w\}$ and that $f(y,\{x,z\},\{z\})=\{x,z\}$. Now consider the path on 6 nodes labelled $uvwxyz$, as well as two paths on three nodes labeled $uvw$ and $xyz$. Since $u,v,y,z$ have the same neighbourhood on $P_6$ as on their corresponding $P_3$ they must be obeying the same policy. Thus, only $w$ and $x$ have a different policy. However, by exhaustive search it can be shown that there exists no protocol that will terminate from every node of the $P_6$, as well as each connected sub-graph. Thus there cannot exist such $u,v,w,x,y,z$. \\

    Consider the digraph $D$ with node set $V=[\kappa]$ and edges $\{(u,v)\in V\times V|u>_*v\}$, the set $S$ in the statement of the lemma corresponds to an independent set of size $c'$. The condition derived in the previous section implies that $D$ is $H$-subgraph free for $H$ the digraph in figure \ref{fig: ramsey}. There exists a $\kappa$ (explicitly $R(8,c')$) such that $D$ must contain either a tournament on $8$ nodes or an independent set on $c'$ nodes. In the latter case the statement holds, in the former consider that there is no way to orient the edges of $K_4$ such that no node has two predecessors. Therefore, if $D$ contains a tournament on $8$ vertices then it contains $H$ as a sub-digraph. Thus, the claim holds for $\kappa>R(8,c')$.
\end{proof}
\begin{figure}
    \centering
    \includegraphics[width=0.7\textwidth]{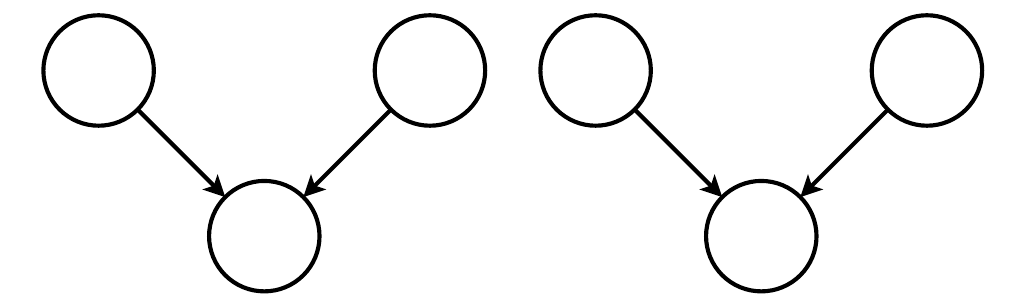}
    \caption{The forbidden sub-digraph used in the argument of lemma \ref{lemma: ramsey}.}
    \label{fig: ramsey}
\end{figure}
With this established, we can strengthen this argument to obtain Amnesiac like behaviour for degree 2.
\begin{lemma}
    \label{lemma:degree 2}
    If there exists $S\subseteq [\kappa]$ such that $|S|>6$ where $\forall u,v \in S$ $u\not>_*v$, then there exists $T\subseteq S$ with $|T|=|S|-3$ such that $\forall u,v,w \in T$, $f(v, \{u,w\},\{u\})=\{w\}$.
\end{lemma}
\begin{proof}
    There are several cases to consider. If there are no triples of the form $u,v,w \in S$ such that $f(v,\{u,w\},\{u\})=\{u,w\}$ then the claim is trivial. Now assume that all such triples include one of the three identifiers $x,y,z \in S$. In this case, we take the remaining $|S|-3$ to be $T$ and the claim holds. Finally, assume that there are two disjoint triples $a,b,c$ and $g,h,i$ with this property. Without loss of generality we can assume that $f(b,\{a,c\},\{a\})=\{a,c\}$ and $f(h,\{g,i\},\{g\})=\{g,i\}$. Taking the path on $6$ nodes labelled $abcghi$, we have that either $b,c,g$ or $c,g,h$ must also be such a triple, or the protocol will not terminate on our path. 
    However, should $bcg$ or $cgh$ be such triples, we can shorten the path by deleting vertices from each end to obtain a labelling that once again does not terminate.
    More concretely, while ever there are only two such triples on the path that are ``facing each other'' they will bounce the message back and forth between each other.
    If there are more than two triples we just remove vertices until only two remain.
    Thus, as $P$ terminates for all graphs and labellings there of, we have that all such triples share one of three identifiers, and so we can take $T$ to be $S\setminus\{x,y,z\}$.
\end{proof}
\begin{corollary}
    If $\kappa\geq R(9,8)$ then $P$ is AF on the set $T$ described in the statement of lemma \ref{lemma:degree 2} up to degree $2$.
\end{corollary}
\begin{proof}
    There are two remaining cases to consider, the behaviour of $b$ and the case where a node receives messages from both of its neighbours.
    In the former case, consider the path on three nodes and broadcast from the central node.
    Since both of the broadcasting node's neighbours are leaves and by assumption will not return the message, for broadcast to be correct the initial node must sent to both of its neighbours.\\
    On the other hand, take $C_3$ labelled with $u,v,w \in T$.
    If $u$ initiates the broadcast, it will receive the message from both $v$ and $w$ three rounds later.
    If it sends to both $v$ and $w$, the process will repeat indefinitely.
    However, if it sends to only one of the two, the message will then circle forever.
    Thus, it must send to neither $v$ or $w$.\\
    Therefore, the claim holds.
\end{proof}
We now show that this behaviour extends to subcubic graphs.
\begin{lemma}
    \label{lemma: degree 2 to 3}
    Let $T \subseteq [\kappa]$ be a set of labels such that $|T|\geq 6$ and $P$ is AF up to degree $2$ on $T$, then $P$ is AF up to degree $3$ on $T$.
\end{lemma}
\begin{proof}
    We will construct a set of subcubic graphs such that only AF behaviour will lead to broadcast and termination for $P$ when they are labelled from $T$.\\
    Our first contestant and the base for much of this argument is the star $S_3$, labelled with $u,v,w,x \in T$ with $u$ the centre.
    Since $P$ is AF on $T$ up to degree $2$, neither $v$, $w$ nor $x$ will ever return a message to $u$.
    Thus for broadcast at $u$ to be successful, we must have that $b(u, \{v,w,x\})=\{v,w,x\}$.
    Furthermore, if broadcast is initiated at a leaf, $u$ must pass the message to both other leaves, thus $\{w,x\}\subseteq f(u,\{v,w,x\},\{v\})$.
    We can then repeat this argument substituting all possible arrangements of $T$ to obtain this rule for all quadruples.
    In fact we will assume this occurs for each graph labelling pair we consider for the remainder of this proof.\\
    Now consider, the same $S_3$ but with $v$ and $w$ connected by a new edge to form a paw graph.
    Should broadcast be initiated at $x$, then after four rounds $u$ will receive a message from both $v$ and $w$.
    If $\{v,w\} \subseteq f(u,\{v,w,x\},\{v,w\})$ then the same state will be reached three rounds later leading to non-termination.
    However, this does not directly preclude an asymmetric policy where $u$ sends a message back to either $v$ or $w$, but not both.  
    Without loss of generality, assume $v \in f(u,\{v,w,x\},\{v,w\})$ but $w \notin f(u,\{v,w,x\},\{v,w\})$. 
    There are now two possibilities, either $f(u,\{v,w,x\},\{w\})=\{v,w,x\}$ or $f(u,\{v,w,x\},\{w\})=\{v,x\}$. 
    In both cases we obtain a repeating sequence and so neither terminate. 
    Thus, $v,w, \notin f(u,\{v,w,x\},\{v,w\})$.\\
    We now take our original $S_3$ and attach a new node $y \in T$ to $v$ and $w$.
    Then if the broadcast is initiated at $y$, $u$ will only receive the message once. Thus as it is a bridge, it must pass the message to $x$ the first time it receives the message. Thus, $f(u,\{v,w,x\},\{v,w\})=\{x\}$ and we only need to determine the policy when receiving from one or all of its neighbours.\\

    Unfortunately, the case for receiving a single message is slightly more complex and will require we make use of a communication graph with two nodes of degree $3$.
    Fortunately, the previous policies apply to all quadruples and so the behaviour for both nodes is partially determined.
    Let $u,v,w,x,y,z \in T$, we then wlog. have three possibilities to consider.
    \begin{itemize}
        \item $x \notin f(u,\{v,w,x\},\{x\})$ and $u\notin f(x,\{u,y,z\},\{u\})$: This is the policy of Amnesiac Flooding.
        \item $x \in f(u,\{v,w,x\},\{x\})$ and $u\notin f(x,\{u,y,z\},\{u\})$: Consider the first graph and labelling in ~\cref{fig: diamond and h}, the protocol will not terminate broadcasting from $u$.
        \item $x \in f(u,\{v,w,x\},\{x\})$ and $u\in f(x,\{u,y,z\},\{u\})$: Consider the first graph and labelling in ~\cref{fig: diamond and h}, the protocol will not terminate broadcasting from $u$ or $x$.
    \end{itemize}
    \begin{figure}
        \centering
        \includegraphics[width=\linewidth]{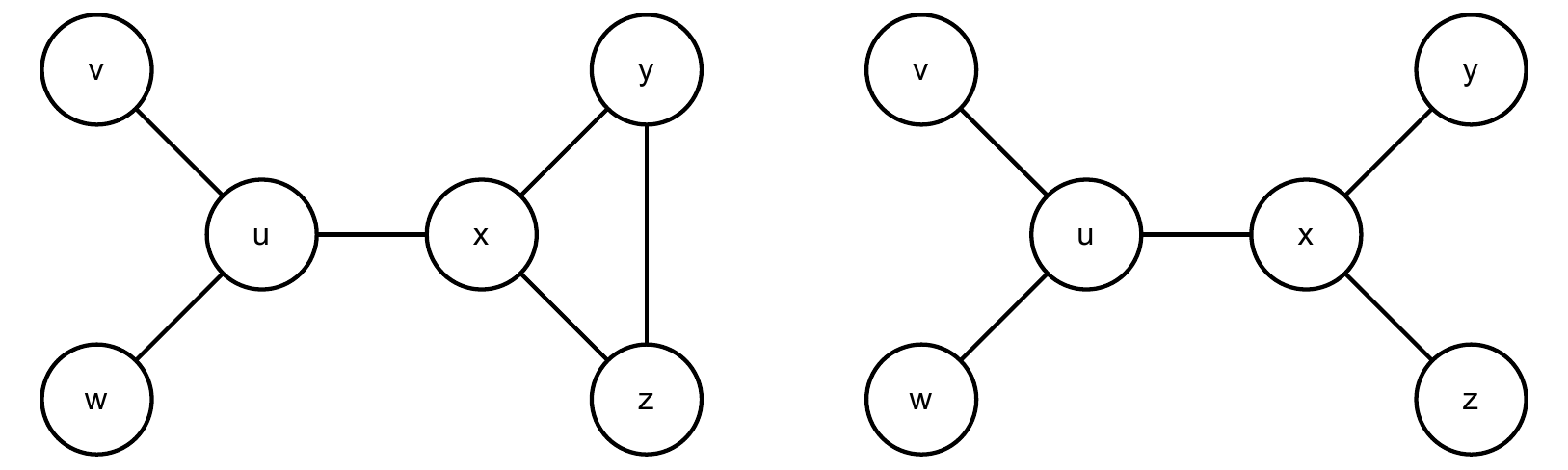}
        \caption{Left: The graph used in the proof of lemma \ref{lemma: degree 2 to 3} to exclude the case $x \in f(u,\{v,w,x\},\{x\})$ and $u\notin f(x,\{u,y,z\},\{u\})$. Right: The H graph used in the proof of lemma \ref{lemma: degree 2 to 3} to exclude the case $x \in f(u,\{v,w,x\},\{x\})$ and $u\in f(x,\{u,y,z\},\{u\})$.}
        \label{fig: diamond and h}
    \end{figure}
    Finally, we must deal with the case where a degree $3$ node receives from all of its neighbours.
    First, consider $K_{2,3}$ with the nodes of degree $3$ labelled $u$ and $y$, with the degree $2$ nodes labelled $v,w,x$.
    If $f(u,\{v,w,x\},\{v,w,x\}) \notin \{\emptyset, \{v,w,x\}\}$ then the protocol will not terminate when broadcast is initiated at $y$.
    However, if $f(u,\{v,w,x\},\{v,w,x\})=\{v,w,x\}$ then the protocol will not terminate on the diamond graph (degree 2 nodes labelled $v$ and $w$, while degree $3$ nodes are labelled $u$ and $x$) when broadcast is initiated at $u$.
    Therefore, $f(u,\{v,w,x\},\{v,w,x\})=\emptyset$.
\end{proof}
For ease of reference we will combine the following results into the more readable statement below.
\begin{lemma}
    There exists $T\subseteq [\kappa]$ such that $|T|\geq 6$ and $P$ is AF on $T$ up to degree 3.
\end{lemma}
With this established we can extend the argument to all labels.
While technically the $T$ referred to in the previous lemma is not necessarily $[|T|]$, all labels within $[\kappa]$ can be reordered arbitrarily as they are available for all graphs.
For this reason and ease of presentation we will simply take $[T]=[|T|]$.
\begin{lemma}
    \label{lemma: id induction}
    $P$ is AF on $\mathbb{N}$ up to degree $3$.
\end{lemma}
\begin{proof}
    Our argument will be by induction, we will show that if $P$ is AF on $[m]$ up to degree $3$, it is AF on $[m+1]$ up to degree $3$. We shall do this by constructing communication graphs where only a few agents may deviate from an AF policy and then showing that none of these policies are correct and terminating. Throughout we will focus on the identifiers $v=m+1$ and $w,x,y,z,a,b,c,d \in [m]$ however our result relies on considering all possible permutations of the $[m]$ labels on our graphs. We note here that $c$ and $d$ are essentially dummy nodes to bound a path of arbitrary size (and possibly trivial) which contains all of the unused labels, thus they do not violate our base case of 6 nodes.

    First consider the extended paw graph depicted and labelled as in figure \ref{fig:PawBase}. The only two nodes that can have non-AF policies in this scenario are $v$ and $w$ and all others must be using AF policies. Since, the long path starting at $y$ contains only identifiers from $[m]$ and has maximum degree $3$, any broadcast initiated outside of the path will never lead to messages leaving that path back into the main graph. Thus, we can essentially ignore it. Then we have eight main cases to consider, namely each combination of $f(v,\{w\},\{w\})=\emptyset$ or $\{w\}$, $f(w,\{v,x\},\{v\})=\{x\}$ or $\{v,x\}$, and $f(w,\{v,x\},\{x\})=\{v\}$ or $\{v,x\}$. If $f(v,\{w\},\{w\})=\{w\}$ and the other two are as in Amnesiac Flooding, then the message will simply bounce back along the extended paw forever. Similarly, if $f(w,\{v,x\},\{x\})=\{v,x\}$ and $f(v,\{w\},\{w\})=\emptyset$, the same will occur. 
    There are in fact only two cases where termination will occur. Firstly, all three can be as in Amnesiac Flooding. Secondly, $f(v,\{w\},\{w\})=\emptyset$, $f(w,\{v,x\},\{v\})=\{v,x\}$, and $f(w,\{v,x\},\{x\})=\{v\}$ also terminates. Full sketches of non-termination are given in figure \ref{fig:Paw Non-Termination} for each of the remaining three cases. By permuting the identifiers other than $v$, we obtain that for all $\alpha,\beta \in [m]$ we must have that $f(v,\{\alpha\},\{\alpha\})=\emptyset$ and $f(\alpha,\{v,\beta\},\{\beta\})=\{v\}$.\\ \\
    \begin{figure}
        \centering
        \includegraphics[width=0.6\linewidth]{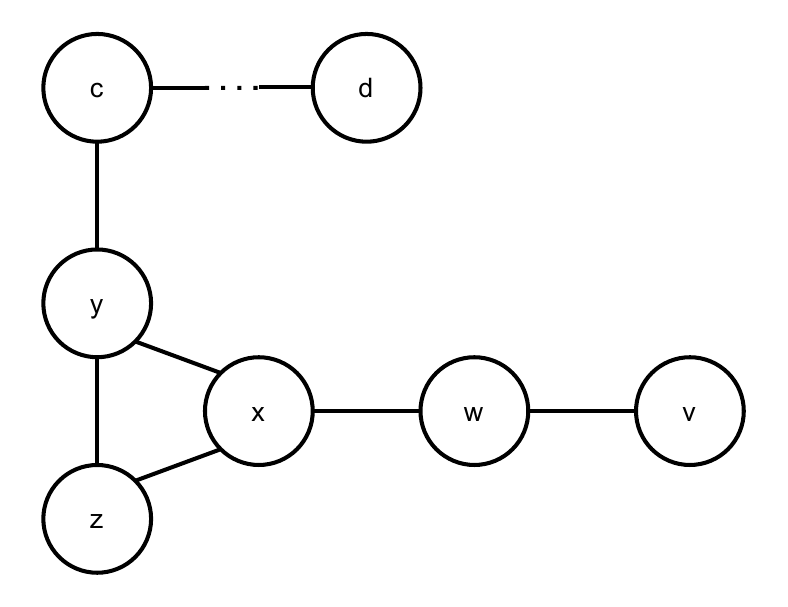}
        \caption{The extended paw graph used in the proof of lemma \ref{lemma: id induction}. Note that the path from $c$ to $d$ represents all identifiers from $[m]$ that are not included in the main body of the paw graph.}
        \label{fig:PawBase}
    \end{figure}
    \begin{figure}
        \centering
        The case $f(w,\{v,x\},\{v\})=\{v,x\}$:\\ \includegraphics[width=\linewidth]{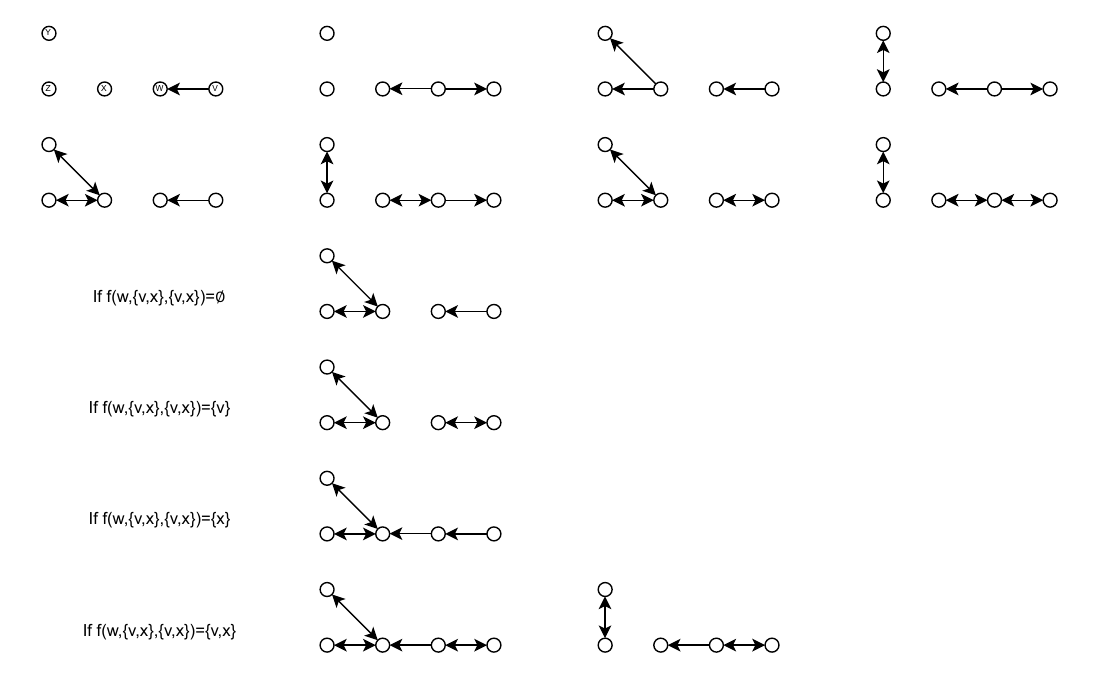}\\
        The case $f(w,\{v,x\},\{x\})=\{v,x\}$:\\ \includegraphics[width=\linewidth]{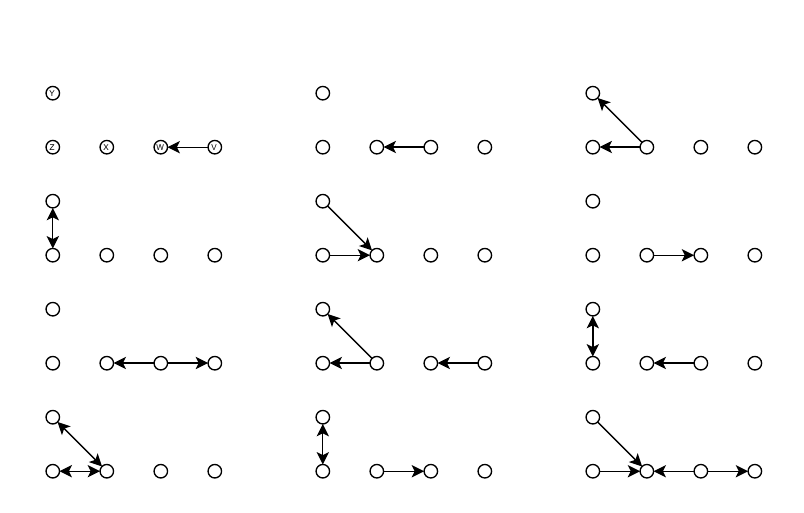}\\
    \end{figure}
    \begin{figure}
    \ContinuedFloat
    \centering
        The case $f(w,\{v,x\},\{v\})=\{v,x\}$ and $f(w,\{v,x\},\{x\})=\{v,x\}$:\\ \includegraphics[width=\linewidth]{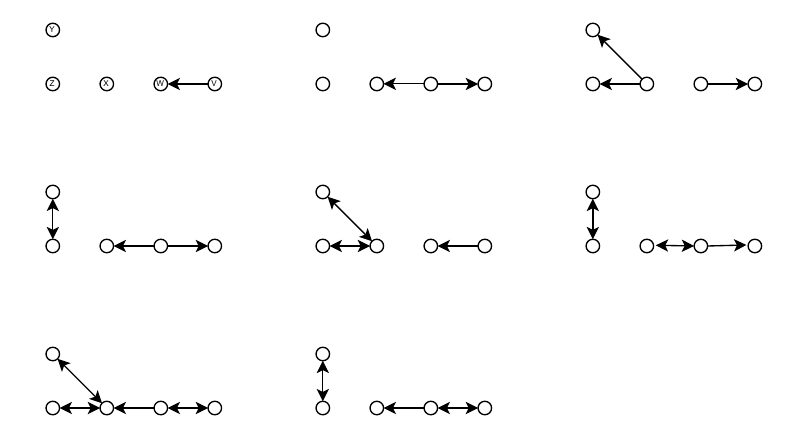}\\
        \caption{A non-terminating broadcast for the cases $f(v,\{w\},\{w\})=\{w\}$ and either: $f(w,\{v,x\},\{v\})=\{v,x\}$, $f(w,\{v,x\},\{x\})=\{v,x\}$ or both, respectively. Message traces are given up to either the final unique state or the first repeat.}
        \label{fig:Paw Non-Termination}
    \end{figure}
    We now need to show that $f(\alpha,\{v,\beta\},\{v\})=\{\beta\}$ as well. Consider a path on $m+1$ nodes labelled $wxvyzc...d$, where the section $c...d$ contains all remaining labels from $[m]$. Since, only $v,w,y$ could have policies different to those of Amnesiac Flooding in this setting, a message can only be sent from outside the subsection $v,w,y$ once, and so we can treat this structure like $P_3$. Similar to the path on three nodes, without loss of generality, there are essentially 12 scenarios to consider, which produce behaviours falling into 5 general cases expressed in the table in figure \ref{fig:table}.
    \begin{figure}
        \centering
        \resizebox{\textwidth}{!}{
        \begin{tabular}{|c|c|c|c|}
             \hline
             Terminating & $f(v,\{w,y\},\{w\})=\{y\}$ & $f(v,\{w,y\},\{w\})=\{y\}$ & $f(v,\{w,y\},\{w\})=\{w,y\}$ \\
             Policies & $f(v,\{w,y\},\{y\})=\{w\}$ & $f(v,\{w,y\},\{y\})=\{w,y\}$ & $f(v,\{w,y\},\{y\})=\{w,y\}$\\
             \hline
             $f(w,\{v,x\},\{v\})=\{x\}$& Case 0& Case 1 & Case 1 \\
             $f(y,\{v,z\},\{v\})=\{z\}$& & &\\
             \hline
             $f(w,\{v,x\},\{v\})=\{v,x\}$& Case 1 & Case 2 & Non-Terminating \\
             $f(y,\{v,z\},\{v\})=\{v\}$& & & \\
             \hline 
             $f(w,\{v,x\},\{v\})=\{v\}$& Case 1 & Non-Terminating & Non-Terminating\\
             $f(y,\{v,z\},\{v\})=\{v,z\}$& & & \\
             \hline 
             $f(w,\{v,x\},\{v\})=\{v,x\}$&Non-Terminating & Case 3 & Case 4 \\
             $f(y,\{v,z\},\{v\})=\{v,z\}$& & & \\
             \hline
        \end{tabular}}
        \caption{The possible terminating cases for the policy discussed in theorem \ref{lemma: id induction} for the path on $m+1$ nodes}
        \label{fig:table}
    \end{figure}
    It remains to construct a communication graph and labelling such that for each case $1$-$4$ termination fails.\\
    For case $1$, there is exactly one of $x,v,y$ which implements a policy different to Amnesiac Flooding, we denote it $\gamma$.
    Furthermore, without loss of generality there exists at least one of its neighbours $\beta$ such that upon receiving a message from $\beta$ it sends messages to both of its neighbours.
    We take the path $wxvyz$, which contains either the segment $\beta \gamma$ or $\gamma \beta$.
    In the former case, we attach $a$ and $b$ to $w$ to form a $C_3$ on $abw$ and the path $c...d$ to $z$, forming an extended paw graph.
    In the latter case, we instead attach $a$ and $b$ to $z$ and the path on $c...d$ to $w$ to again form an extended paw.
    Now consider broadcast initiated at the solitary leaf node.
    The message will reach $\gamma$ and then bounce back and forth between $\gamma$ and the $C_3$ for ever.

    In case $2$, upon receiving from $v$, $w$ will send to both of its neighbours, as will $v$ upon receiving from $y$.
    This means that they both respond in the same direction on the path.
    We are therefore able to construct an extended $C_5$ as in figure \ref{fig:C5}, on which broadcast will not terminate when initiated at $z$.\\
    \begin{figure}
        \centering
        \includegraphics[width=0.4\linewidth]{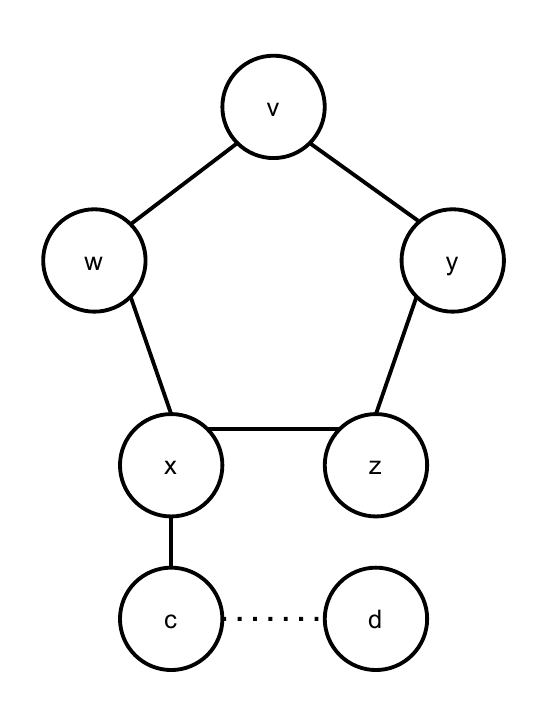}
        \includegraphics[width=0.5\linewidth]{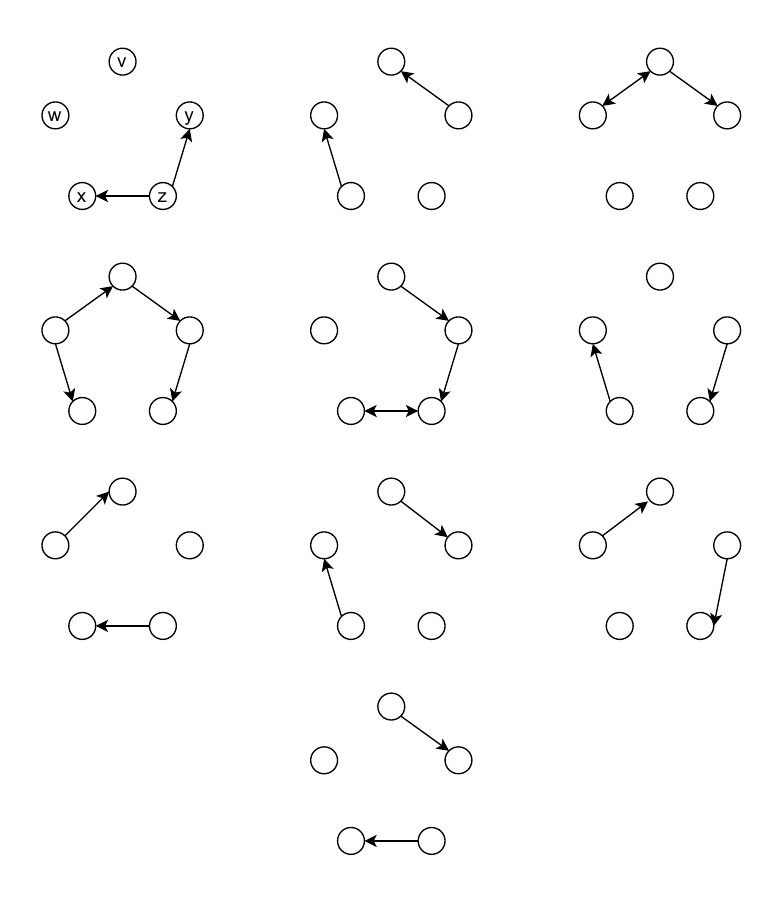}
        \caption{A communication graph for Case 2 where the path $c...d$ contains all remaining identifiers from $[m]$ and a non-terminating message sequence initiated at node $z$. The message sequence is given to the final unique configuration.}
        \label{fig:C5}
    \end{figure}
    In case $3$, both $w$ and $y$ will respond to $v$ however $v$ responds only to $y$.
    We construct a paw graph as in Case $1$. 
    In fact, we can use the same paw graph for case $4$, as both will fail to terminate when broadcast is initiated at $z$.
    For, the precise messages see figure \ref{fig: Cases3 and 4}.\\
    \begin{figure}
        \centering
        \includegraphics[width=\linewidth]{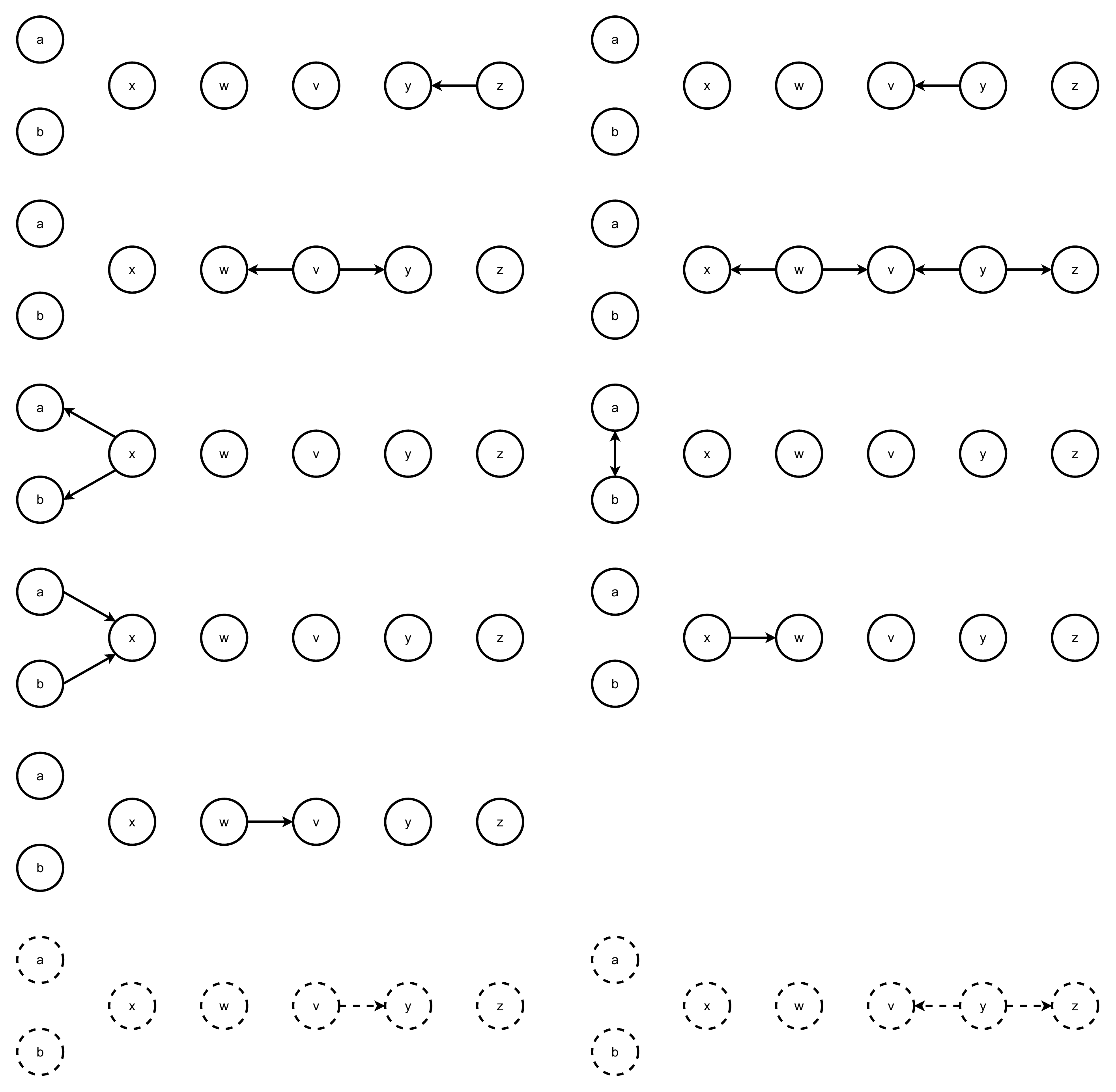}
        \caption{Non-terminating messages sequences for cases $3$ and $4$ in lemma \ref{lemma: id induction}. Both behave identically for the first 9 iterations, however case $3$ takes two further iterations to repeat a step (indicated by dashed lines). Note that path $c...d$ has been omitted here but would be attached to $z$.}
        \label{fig: Cases3 and 4}
    \end{figure}
    Thus as we have obtained non-termination for all cases other than case $0$, we must have for all $\alpha,\beta,\gamma \in [m+1]$ that $f(\alpha, \{\beta, \gamma\}, \{\beta\})=\{\gamma\}$.
    Furthermore, we must have that $b(\alpha, \{\beta, \gamma\})=\{\beta,\gamma\}$ as otherwise the protocol would not terminate on a cycle.
    This also implies that $f(\alpha, \{\beta,\gamma\},\{\beta,\gamma\})=\emptyset$ to ensure termination on odd cycles when broadcast is initiated at $\alpha$.
    Finally, we must have that $f(\alpha,\{\beta\},\{\beta\})=\emptyset$ as $\alpha$ could be on the end of an extended paw.
    Thus, we have that $P$ is AF on $[m+1]$ up to degree $2$.\\

    Fortunately, we can reuse much of the previous argument for obtaining degree $3$ behaviour.
    Excluding the diamond graph and the $K_{3,2}$ graph, we can use all of the constructions in \cref{lemma: degree 2 to 3} by simply adding a path $c...d$ to a node of degree at most $2$ which is not adjacent to $v$.
    However, we can slightly alter the diamond and $K_{3,2}$ graph to permit the addition of such path to a node of degree at most $2$ which is not adjacent to $v$ no matter where $v$ appears in the labelling.
    These are given explicitly in figure \cref{fig:degree 3 but longer}.
    Note that while, these do require $8$ nodes to construct, we can use the original constructions from \cref{lemma: degree 2 to 3} until we have more than $\kappa$ nodes.
    \begin{figure}
        \centering
        \includegraphics[width=\linewidth]{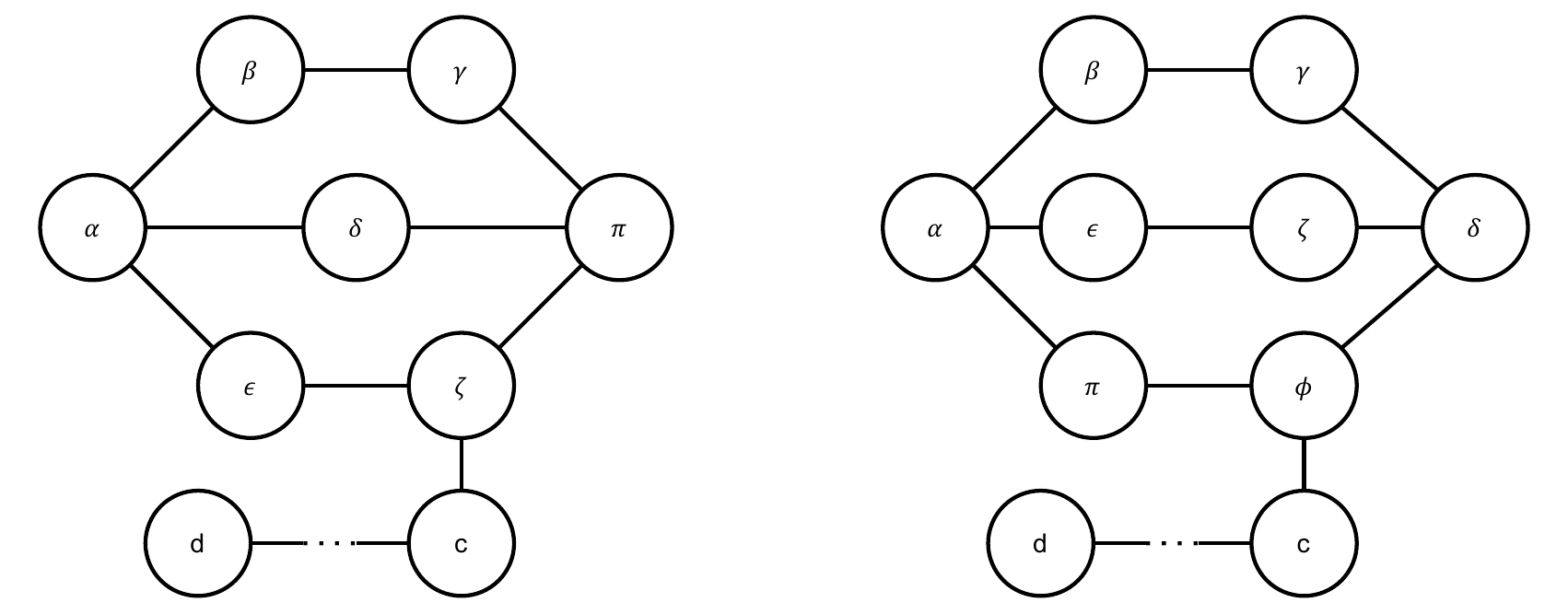}
        \caption{The two graphs used to determine the policy for identifiers of degree 3 in lemma \ref{lemma: id induction}. Left: The analogue to the diamond graph. Right: The analogue to $K_{3,2}$}
        \label{fig:degree 3 but longer}
    \end{figure}
    Thus, repeating the argument from \cref{lemma: degree 2 to 3}, we have that $P$ is AF on $[m+1]$ up to degree $3$.
    Since, $P$ is AF on $T$ up to degree $3$, by induction we have that $P$ is AF on $\mathbb{N}$ up to degree $3$.
\end{proof}
Thus painfully we have that $P$ is indistinguishable from Amnesiac Flooding on subcubic graphs.
It may reassure the reader that we will not have to repeat this process for every maximum degree as $3$ is something of a magic number. 
In particular, once nodes behave up to degree $3$ we can construct binary tree like structures to simulate high degree nodes with lower ones.
With this tool under our belt, we are ready to prove theorem \ref{thm: AF is unique}.
\begin{proof}[Proof of Theorem \ref{thm: AF is unique}]
    By lemma \ref{lemma: id induction} we have that $P$ is AF up to degree $3$ on $\mathbb{N}$.
    It is immediate that for all $u \in \mathbb{N}$, $S \subset \mathbb{N}$ $b(u,S)=S$ as $u$ could be at the centre of a star with one of its leaves replaced by a long path.
    In this case, $u$ will never receive a message again and so must send to all of its neighbours in the first round.

    For an arbitrary label $u \in \mathbb{N}$ and degree $k\in \mathbb{N}$ we can show that $u$ must behave as though it is implementing Amnesiac Flooding if it is at a node of degree $k$.
    Take $S \subset \mathbb{N}$ such that $|S|=k$, we will show that for all non-empty $T\subseteq S$, $f(u,S,T)=S \setminus T$ via induction on the size of $T$.

    Beginning with our base case.
    If $u$ receives from a single neighbour we can construct a pair of graphs (special cases of our general construction) that enforce the AF policy. 
    Consider the tree in figure \ref{fig:OneMessageDegreeK}, if $u$ receives a message from $v$ it must send to at least $x_1,...,x_{k-1}$ as $u$ will receive messages in only a single round. 
    This follows as the rest of tree will use AF policies and since the graph is bipartite each node will be active only once. 
    Thus, the only two options for $f(u,S,\{v\})$ are $S$ or $S\setminus \{v\}$.\\ \\
    Now consider the second graph from figure \ref{fig:OneMessageDegreeK} and a broadcast initiated at $u$. 
    We can see that if $f(u,S,\{v\})=S$ then the cycle and $u$ will simply pass a message back and forth forever. 
    Thus, $f(u,S,\{v\})=S\setminus\{v\}$.
    \begin{figure}
        \centering
        \includegraphics[width=\linewidth]{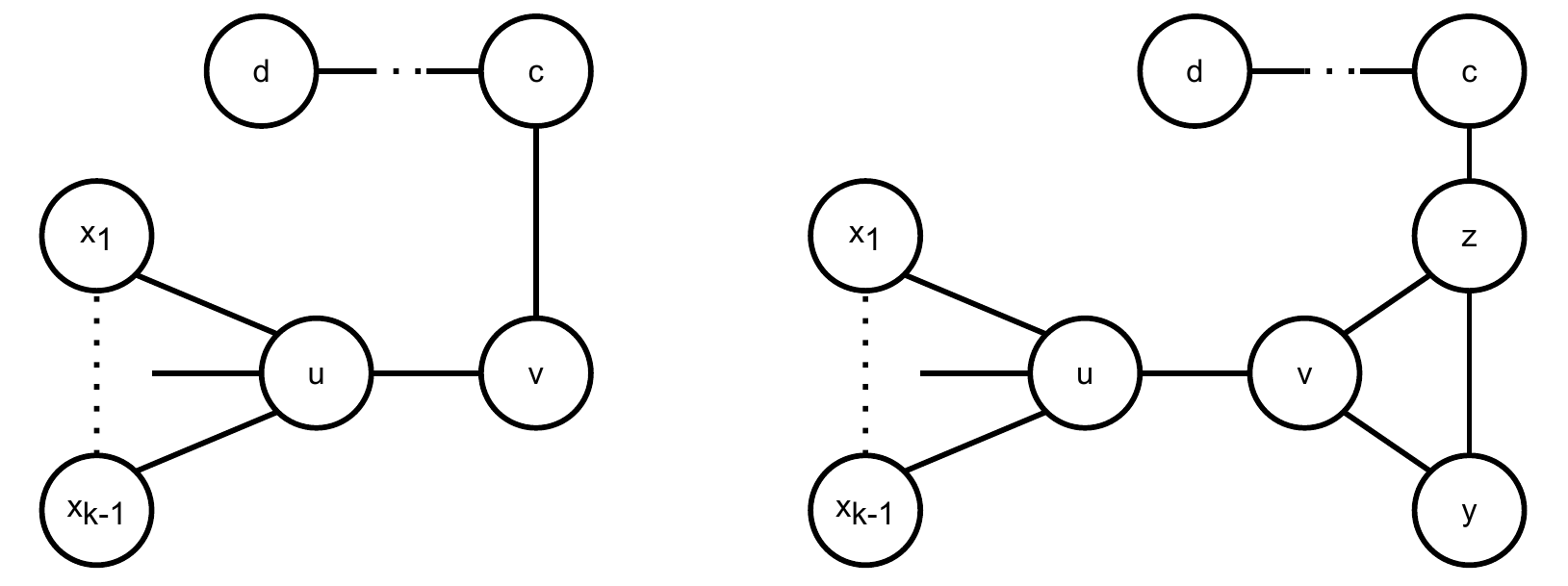}
        \caption{Two graphs used in the proof of theorem \ref{thm: AF is unique} to determine identifiers' response to receiving only a single message. In both cases $c...d$ is a path containing all unused identifiers from $[m]$ where $m$ is the highest id used in labelling the body of the graph. Left: A tree that forbids sending to too small a subset of neighbours. Right: A graph that forbids sending a message to all neighbours.}
        \label{fig:OneMessageDegreeK}
    \end{figure}
        We can now generalize this construction and perform our induction. 
    Assume that for any non-empty subset $T$ of $S$ of size at most $q$, $f(u,S,T)=S\setminus T$. 
    For the sake of contradiction assume that this is not true for $V$ where $T\subset V \subseteq S$ with $|V|=q+1$. 
    Thus either $f(u,S,V)\cap V\neq \emptyset$ or $f(u,S,V)\cup V\neq S$. \\ 
    In the first, case let $W=V\cap f(u,S,V)$, and for some ordering label the elements of $W$: $w_1,...,w_r$ and the elements of $V\setminus W$: $v_1,...,v_{q+1-r}$.\\ 
    If $W=V$, then we construct a communication graph of the form in figure \ref{fig: BaseMany-MessagesDegreeK} and consider a broadcast initiated at $u$. 
    As every node other than $u$ in the graph has maximum degree $3$, the rest of the graph will use policies indistinguishable from Amnesiac Flooding. 
    Thus, the messages will travel in only one direction through the binary tree and so at the time $w^*$ receives a message it will be the only message on the graph. 
    This message will then be passed onto the cycle where it will circulate, before being passed back onto the binary tree in the opposite direction. 
    Again the policy of all nodes other than $u$ is indistinguishable from Amnesiac Flooding and so when $u$ next receives a message, it receives from all of $V$ and these are the only messages (outside of the path from $c$ to $d$).
    We therefore have a repeating sequence as $u$ will forward the message back to every node of $W$.
    \begin{figure}
        \centering
        \includegraphics[width=\linewidth]{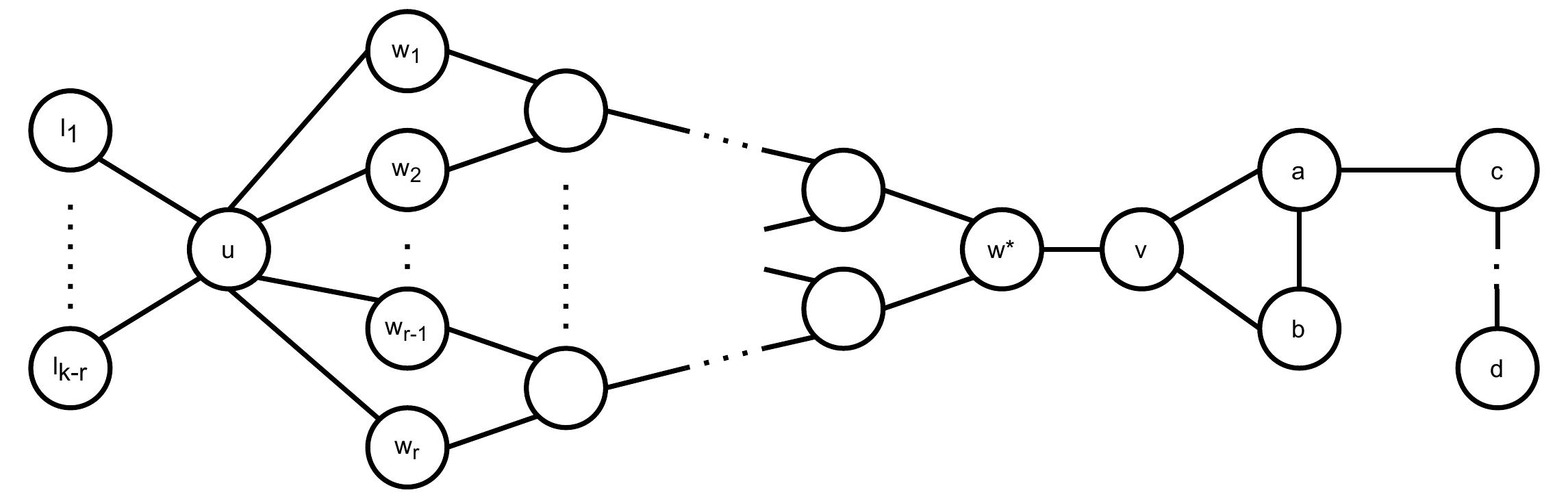}
        \caption{A graph used in the proof of theorem \ref{thm: AF is unique}. The graph consists of a star centred at $u$ with its leaves partitioned into two sets of size $r$ and $k-r$. The leaves in the set of size $r$ are connected by a binary tree of depth $\lceil \log_2{r}\rceil$ with root $w^*$ to a cycle, which in turn is connected to a path. Note that the path $c...d$ contains all identifiers from $[m]$ not used in the labelling of the rest of the graph, where $m$ is the largest $id$ present.}
        \label{fig: BaseMany-MessagesDegreeK}
    \end{figure}\\
        If $W\subset V$, then we construct a slightly different communication graph (see figure \ref{fig:Many-MessagesDegreeK}).
    This graph instead partitions $S$ into $W$, $V \setminus W$ and $S \setminus V$, with separate binary trees for $W$ and $V \setminus W$.
    Here we consider a broadcast initiated at $a$.
    As every node other than $u$ has maximum degree $3$ it must make the same forwarding decisions as Amnesiac Flooding. 
    Thus, when $u$ first receives a message it will receive a message from all identifiers in $V$ and there will be no other messages in the body of the graph.
    Then by assumption $u$ will send a message to some portion of $\{l_1,...,l_{k-q-1}\}$ as well as all of $W$. 
    The leaves will not respond and messages will travel only upwards in the binary tree with $W$ as leaves and $w^*$ as its root. 
    Thus, $a$ will next receive a message only from $w^*$ and since it makes the same decision of Amnesiac Flooding it will send to $v^*$ and $c$. 
    Until $u$ receives a message the only messages in the body of the graph will be those travelling down the binary tree with $V\setminus W$ as its leaves and they will all arrive at $u$ simultaneously. 
    Thus, $u$ will then receive only from $V\setminus W$. Since $|V\setminus W|<|V|=q+1$ by assumption $u$ must make the same decision as Amnesiac Flooding and so will send messages to all of $W$ and its leaves. 
    This creates a repeating sequence and so we have non-termination. 
    Thus, since $f(u,S,V)\cap V\neq\emptyset$ always allows us to construct a communication graph with a non-terminating broadcast, we must have that $f(u,S,V)\cap V =\emptyset$.
    \begin{figure}
        \centering
        \includegraphics[width=\linewidth]{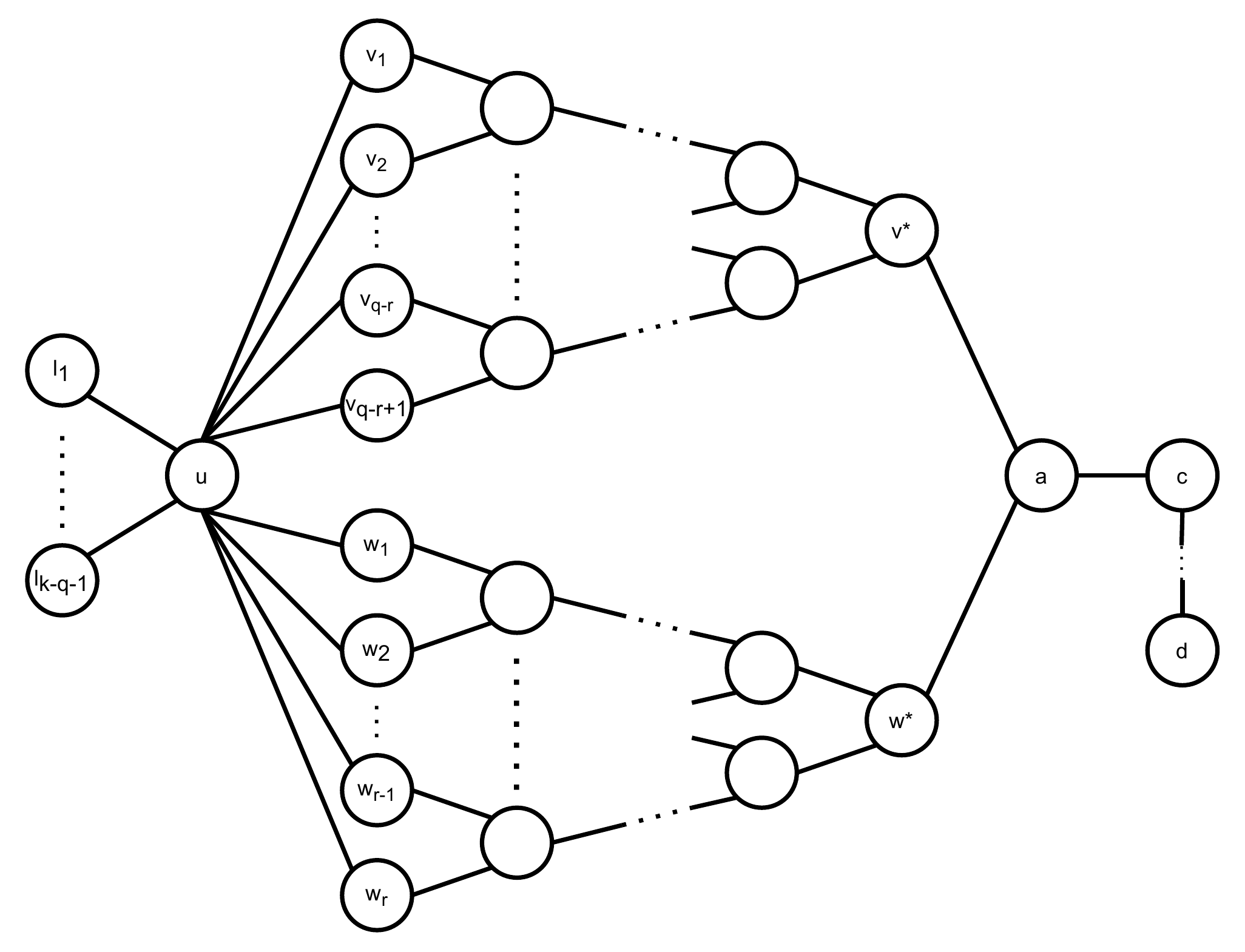}
        \caption{A graph used in the proof of \cref{thm: AF is unique}. The graph consists of a star centred at $u$ with its leaves partitioned into three sets of size $r$, $q-r+1$ and $k-q-1$. The  leaves in the first two sets are then each joined to a single node labelled by $w*$ and $v*$ respectively by binary trees of depth $\lceil \log_2{q}\rceil$. The single nodes are connected to a node labelled $a$ which is the start of a path $ac...d$ that contains all identifiers in $[m]$ not used in the labelling of the body of the graph where $m$ is the largest id used in that labelling.}
        \label{fig:Many-MessagesDegreeK}
    \end{figure}
    Now consider the communication graph from figure \ref{fig:Many-MessagesDegreeK} again but with $W$ an arbitrary strict subset of $V$.
    Since $f(u,S,V)$ does not contain any id from $V$ and none of $\{l_1,...,l_{k-q-1}\}$ will send a message back to $u$, $u$ sends messages in only a single round. 
    Therefore, $u$ must send to all of $\{l_1,...,l_{k-q-1}\}$ otherwise some would not receive the message (and so the protocol would not implement broadcast correctly). 
    Thus, $f(u,S,V)=S\setminus V$ and so we have our contradiction.

    Therefore, if $f(u,S,T)=S\setminus T$ for $S\subset \mathbb{N}$ where $|S|=k$ and all $T\subseteq S$ such that $0<|T|\leq q$, then $f(u,S,V)=S\setminus V$ for $V\subset S$ such that $|V|=q+1$. 
    Thus, by induction since we know this to be true for all $u$ and $S$ when $q=1$ it must hold for all $q<|S|$.

    This gives our claim, as for every $u$ and $k$ we can apply this argument and show that for any $k \in \mathbb{N}$, $P$ is AF up to degree $k$.
\end{proof}
\section{Uniqueness: Existence of Protocols in the relaxed cases}
\label{apx: relaxed}
In this appendix, we present protocols for each of the relaxations discussed in \cref{subsec: relaxation}
\subsection{Protocol with limited statelessness}
For the protocol with limited statelessness, we simply exploit the ability of the protocol to identify whether it is on a star.
The protocol can be stated as follows:
\begin{definition}[Neighbourhood 2 flooding]
    \emph{Neighbourhood 2 Flooding} is the protocol defined by the following rule:
    upon receiving a message, if a node has a neighbour of degree $2$ or higher, \emph{or} has only a single neighbour, it implements Amnesiac Flooding.
    Otherwise, upon receiving a message it sends to all of its neighbours.
\end{definition}
\begin{proposition}
    Neighbourhood 2 flooding is correct and terminates in finite time.
\end{proposition}
\begin{proof}
    The protocol is automatically correct and terminating for non-star graphs by the correctness of Amnesiac Flooding.
    If the protocol is on a star, then either the broadcast is initiated at the centre or a leaf.
    On either the first round (if at the centre) or the second (if at a leaf), the centre sends to all of the leaves.
    The leaves all have a neighbour of degree $2$ or more and so do not respond.
    Thus all nodes have received the message and no further messages are sent.
    In the special case of $K_2$ we simply treat the initiating node as the centre, and the same argument holds.
\end{proof}
This implies that while the ability to retain information about the identifiers in a nodes neighbourhood, does not permit anything other than Amnesiac Flooding, even slightly more information does as Neighbourhood 2 flooding produces a different sequence of messages on star graphs.
However distinct behaviour is obtained only for graphs of constant diameter.
In general, a similar construction (Neighbourhood k flooding) can be produced if nodes know the structure of their $k$th order neighbourhood and send to all neighbours upon receiving a message if and only if they are on a tree, can see all of its leaves and have the highest ID out of all nodes.
This still obtains only a linear maximum diameter in the number of hops nodes aware are of.
It remains open whether an protocol in this setting can be found with behaviour distinct from Amnesiac Flooding on graphs of unbounded diameter.
\subsection{Protocol with non-determinism}
For a non-deterministic relaxation we assume each agent has access to a single bit of uniform randomness per round, and otherwise work in the setting of \cref{def: functions}.
We note that while these bits must be independent between rounds, they can be either shared or independent between agents.
With this established we present the following protocol which is correct and terminates in finite time with probability $1$.
\begin{definition}[Random Flooding]
    \emph{Random Flooding} is the protocol defined by the following rule: on each round, if a node has a random bit equal to $0$ it implements Amnesiac Flooding. Otherwise upon receiving a message it sends to all neighbours.
\end{definition}
\begin{proposition}
    Random Flooding correctly achieves broadcast in diameter rounds and terminates in finite time with probability $1$.
\end{proposition}
\begin{proof}
    Regardless of the random bits, messages will propagate down all edges and all agents will receive a message along the shortest path from the initiator node.
    Thus, the protocol is correct and achieves broadcast in exactly $D$ rounds.\\
    Consider the random variable $X(t)$ which for $t \in \mathbb{N}$ takes the value $1$ if on rounds $t, t+1,...,t+D$ all nodes receive $1$ as a random bit, but at round $t+D+1$ all nodes receive $0$, and $0$ otherwise.\\
    We now make the following claim, if for some $t \in \mathbb{N}$, $X(t)=1$, then the protocol has terminated by time $t+D+2$. 
    If there are no messages at time $t$ then, Random Flooding cannot send any new messages and so the protocol has already terminated.
    Otherwise, there must exist at least one node receiving a message at time $t$.
    Since this node has a random bit of $1$, it will send to all of its neighbours.
    Repeating this notion, at time $t+D$ all nodes will be sending to all of their neighbours.
    However, at time $t+D+1$ all agents implement Amnesiac Flooding, and so upon receiving from all of their neighbours terminate, giving our claim.\\
    The event $X(t)=1$ occurs with positive probability $\Omega(2^{-n^2})$ and $X(t)$ is independent of $X(t+D+2)$ for all $t \in \mathbb{N}$.
    Let $Y=\cap_{i\geq 1}\{X(i(D+1))=0\}$, since $Y$ is the limit event of a sequence of monotonically decreasing events and $\lim_{m\to\infty}\mathbb{P}(\cap_{1\leq i \leq m}\{X(i(D+1))=0\})\to 0$, $\mathbb{P}(Y)=0$.
    Thus, Random Flooding terminates in finite time with probability $1$.
\end{proof}
In fact, we can weaken the setting further such that the ``random'' bits are drawn from some sequence with only the fairness guarantee that $X(t)=1$ occurs almost surely for finite $t$.
\subsection{Protocols that can read messages/send multiple messages}
We lump these two together, as we will demonstrate that only one bit of information need be propagated with the message.
This could be in the form of readable header information, but can equally be encoded in the decision to send one or two messages.
Without specifying, we take the setting of \cref{def: functions} but where agents have access to a single bit of information determined by the initiator when broadcast begins and which is available to them on any round where they receive a message.
We will exploit this single bit of information to encode whether the initiator node is a leaf or not, which permits us to make use of the following subroutine,
\begin{definition}[Parrot Flooding]
    \emph{Parrot Flooding} is the protocol defined by the following rule: if a node is a leaf upon receiving a message it returns it to its neighbour. Otherwise, all non-leaf nodes implement Amnesiac Flooding.
\end{definition}
Parrot flooding has the following useful property, the proof of which we defer.
\begin{lemma}
    \label{lemma: parrot flooding works}
    On any graph $G=(V,E)$ with initiator node $v \in V$ such that $v$ is not a leaf node, parrot flooding is correct and terminates in finite time.
\end{lemma}
This allows us to obtain the following protocol which can be adapted to either the higher bandwidth or readable header relaxation.
\begin{definition}[1-Bit Flooding]
    \emph{1-Bit Flooding} is the protocol defined according to the following rule: When broadcast begins, the initiator picks $0$ for the shared bit if it is a leaf and $1$ otherwise.
    Upon receiving a message, if the shared bit is a $1$ the node implements parrot flooding, and implements Amnesiac Flooding otherwise.
\end{definition}
Therefore, we immediately have.
\begin{proposition}
    On any graph $G=(V,E)$ 1-Bit flooding is correct and terminates in finite time.
\end{proposition}
If we have readable header information, then the single bit is placed in the header.
On the other hand, if the nodes are able to send multiple messages per edge per round, then the initiator may choose to send $2$ messages in the case where it wishes to encode a $1$ and a single one otherwise.
In this case, if nodes always forward the same number of messages then the single bit is propagated and maintained.\\

The remainder of the section is devoted to the proof of \cref{lemma: parrot flooding works}.
First, we formalise parrot flooding in the same manner as Amnesiac Flooding.
\begin{definition}[Parrot Flooding Redefinition]
    Let $S$ be a configuration on $G=(V,E)$ as in Amnesiac Flooding. The \emph{Parrot Flooding} protocol functions as follows: $P_{G,I}(S)=A_{G,I}(S)\cup T$ where $T=\{(v,u)|(u,v) \in S \wedge deg(v)=1\}$.
\end{definition}
We begin with the following observation.
\begin{observation}
    Let $S$ be a configuration of messages on $G=(V,E)$. Since messages are only reintroduced at leaves rather than on cycles or FECs, $P(S)$ is balanced on $G$ if and only if $S$ is balanced on $G$. Formally, this follows from \cref{thm: Byzantine}.
\end{observation}
However, if parrot flooding is initiated at the terminal vertex of a path this configuration is balanced but will not terminate.
Therefore, we require a stronger condition than balance to capture termination under parrot flooding.
Consider a message path of parrot flooding (the definition of which is obtained by simply replacing the operator $A$ with $P$ in definition \ref{def: message paths}).
We immediately inherit \cref{lemma: recurrence,lemma: odd cycles,lemma: even cycles,lemma: even cycles,lemma: EvenSub} as none of them depend on the behaviour of leaves.
Thus any recurrent message $(u,v) \in S$ must have a message path $uv...wlw...uv$ where $l$ is a leaf node or it would be captured by the argument used in the proof of \cref{thm:balance}.
Therefore, in fact there exists a configuration $\hat{S}=P^k(S)$ for some $k \in \mathbb{N}$ such that $(w,l)\in \hat{S}$ where $(w,l)$ is recurrent on $\hat{S}$.
The rest of the argument is similar to that used in the proof of \cref{thm:balance}, and we will construct a suitable notion of balance.
First, however, we must define the additional structure that produces non-termination.
\begin{definition}[Leaf paths]
    A \emph{leaf path} $L=l_1...l_k$ on $G=(V,E)$ is a sequence of nodes from $V$ starting and ending at leaves such that $l_il_{i+1} \in E$ and $l_i\neq l_{i+2}$ unless $i=1$ and $k=3$.\\
    The path representation of $L$ denoted $L_p$ is given by the path $l_1,...,l_k$ where copies of the same node are treated as distinct (if $l_i$ is the $j$th appearance of some node $v \in V$ we replace it with a new node $v_j$).\\
    Let $S\subseteq V^2$ be a configuration on $G$, then the path representation of $S$ with respect to $L$ denoted $S_{L_{p}}$ is given by taking all messages from $S$ only between nodes of $L$ and for each message $(u,v)$ adding the set $\{(u_i,v_j)|u_i \in L_p\wedge v_j \in L_p\}$.
\end{definition}
\begin{definition}
    A config $S$ is \emph{$P$-Balanced} on a leaf path $L$ if for all messages $m$ in $S_{L_p}$, the subset of messages that have an even distance between their head and $m$'s along $L_{p}$ is of even cardinality.
    Let $G=(V,E)$ be a graph and $S\subseteq V^2$ a configuration on $G$. Then $S$ is $P$-Balanced on $G$ if $S$ is balanced on $G$ and $P$-Balanced on all leaf paths of $G$.
\end{definition}
This property is conserved by parrot flooding.
\begin{lemma}
    $P$-Balance is preserved under Parrot flooding with non-leaf external node activation.
\end{lemma}
\begin{proof}
    We have already established that balance is preserved by parrot flooding, so it remains to show that $P$-Balance is conserved on all leaf paths.
    Let $Q$ be a leaf path starting at $u$ and ending at $v$ on $G=(V,E)$ and let $S$ be a configuration on $G$.
    Then $P_G(S)_{Q_{p}}=P_{Q_p,I}(S)$ for some $I \subset V$ such that $u,v \not\in I$, i.e. a step of parrot flooding on $G$ is equivalent to a step of parrot flooding on the leaf path representation with some possible external node activation.
    After a step of parrot flooding, two messages on $Q_p$ will either have the distance between their heads: stay the same, increase by two or decrease by two.
    Upon externally activating a non-leaf node two messages are created, with their heads an even distance apart on $Q_p$.
    Symmetrically, a collision between two messages can only be between two messages with an even distance (explicitly $0$) between their heads and removes both.
    Thus, $P$-Balance is maintained.
\end{proof}
\begin{corollary}
    If $S$ is not $P$-Balanced on $G$ then $\not\exists k \in \mathbb{N}$ such that $P_G^k(S)=\emptyset$.
\end{corollary}
We now show that from a $P$-Balanced configuration no message can take too many steps along a leaf path.
\begin{lemma}
    Let $G=(V,E)$ be a graph and $m$ a message in $S\subseteq V^2$. If $m$ has a message path from $S$ on $G$ that takes $k+1$ consecutive steps along a leaf path of length $k$ then $S$ is not $P$-Balanced on $G$.
\end{lemma}
\begin{proof}
    For contradiction assume such a message path exists on the leaf-path $L$ and $S$ is $P$-Balanced on $G$. Then since $m$ takes a step on $L$ there must exist an odd number of other messages on $L_p$ with an even head distinct to $m$ as $S$ is $P$-Balanced on $L$.
    These messages will collide with $m$ the first time they cross over, which unless they are all blocked will occur before $k+1$ steps occur. 
    However, for them to all be blocked, new messages must be added to the $L_p$ with an even head distance to $m$. 
    But, any new message $m_1$ blocking a potential blocker $m'$ of $m$ must introduce $m_2$ with an even head distance to $m$ which will reach $m$ before $m'$ would have.
    Thus, $m$ will collide with another message on $L$ before taking $k+1$ steps along $L$.
\end{proof}
This gives us what we need to mirror the result of \cref{thm:balance}.
\begin{lemma}
    Let $G=(V,E)$ be a graph and $S\subseteq V^2$ be a configuration of it, $S$ has a recurrent message in parrot flooding if and only if $S$ is not $P$-Balanced on $G$.
\end{lemma}
\begin{proof}
    If $S$ is balanced on $G$ we must have that any recurrent message uses a leaf node in its message-path. Therefore, that message must traverse a full leaf path and then take an additional step.
    By the previous lemma this cannot happen if the configuration is $P$-Balanced on $G$. 
    Therefore, a recurrent message does not exist if $S$ is $P$-Balanced on $G$.

    However, if $S$ is not $P$-Balanced on $G$ then the configuration is non-terminating and so must have a recurrent message.
\end{proof}
\begin{corollary}
    \label{corr: Parrots}
    Parrot flooding terminates on $G$ from a configuration $S$ if and only if $S$ is $P$-Balanced on $G$.
\end{corollary}
\begin{proof}[Proof of \ref{lemma: parrot flooding works}]
    Parrot flooding correctly broadcasts as every node receives a message along the shortest path from a member of the initiator set.
    In the first round, every node in the initiator set sends to all of its neighbours, corresponding to an external activation of any leaf paths that it may lie on.
    This cannot break $P$-Balance and since $\emptyset$ is $P$-balanced, by corollary \ref{corr: Parrots} the protocol must terminate in finite time. 
\end{proof}

\section{Proof of the Fault Sensitivity results}
\label{apx: faults}
In this appendix, we prove the fault sensitivity and extended dichotomy results of section \ref{sec: faults}, using \cref{thm:balance} as our main technical tool.
\subsection{Extended Dichotomy}
Here we present a proof of lemma \ref{lemma: going backwards}.
\begin{lemma}[Lemma \ref{lemma: going backwards} restated]
        Let $G=(V,E)$ be a graph, $S$ a configuration of messages on $G$ and $T=\{u \in V| \forall v \in N(v): (v,u) \in S\}$ the set of sink vertices. Then $A_{T,G}(\overline{A_{\emptyset,G}(\bar{S})})=S$
    \end{lemma}
\begin{proof}[Proof of lemma \ref{lemma: going backwards}]
    Assume the message $(u,v) \in A_T(\overline{A(\bar{S})})$ but $(u,v) \notin S$. There are two cases. The first case is that $u$ is an initiator and therefore in $T$, however by definition the nodes in $T$ are sinks in $\bar{S}$ and so this is impossible. The second is that there exists some message $(w,u) \in \overline{A(\bar{S})}$ and $(v,u) \notin \overline{A(\bar{S})}$. For $(w,u)$ to be in $\overline{A(\bar{S})}$ requires that $(u,w) \in A(\bar{S})$, this requires some $(x,u)$ to be in $\bar{S}$. However, since $(v,u) \notin \bar{S}$ we must have that $(u,v) \in A(\bar{S})$ and so $(v,u)$ must be in $\overline{A(\bar{S})}$. Thus, we have a contradiction and so $A_{T,G}(\overline{A_{\emptyset,G}(\bar{S})})\subseteq S$

    Consider a message $(u,v) \in S$ and thus, $(v,u) \in \bar{S}$. If $v$ is not a sink in $\bar{S}$, either there exists a message leaving $u$ over an edge or an edge incident to $u$ over which no message is sent in either $S$ or $\bar{S}$. In either case, $A(\bar{S})$ will include a message sent over that edge which we denote $(u,w)$. Furthermore, we know that $(v,u) \not\in \overline{A(\bar{S})}$ as that would require $(u,v) \in A(\bar{S})$  which is impossible as $(v,u) \in \bar{S}$. Since $\overline{A(\bar{S})}$ must contain the message $(w,u)$, $(u,v)$ must belong to $A(\overline{A(\bar{S})})$.

    On the other hand, assume that $v$ is a sink in $\bar{S}$, then it is in $T$ and will be an initiator in $A_{T,G}(\overline{A_{\emptyset,G}(\bar{S})})=S$ unless it receives a message in $\overline{A_{\emptyset,G}(\bar{S})}$. For $v$ to receive a message in $\overline{A_{\emptyset,G}(\bar{S})}$ requires that $v$ sends a message in $A_{\emptyset,G}(\bar{S})$. However, $v$ is a sink in $\bar{S}$ and so cannot send any messages in $A_{\emptyset,G}(\bar{S})$. Thus we have a contradiction. Therefore, $S\subseteq A_{T,G}(\overline{A_{\emptyset,G}(\bar{S})})$.
\end{proof}
\subsection{Fault Sensitivity}
We begin with the single message failure case, restated below.
    \begin{theorem}[Single Message Failure]
    \label{thm:message drop}
        If single-source Amnesiac Flooding experiences a \emph{single message drop failure} for the message $(u,v)$ then it fails to 
       terminate
        if either:
        \begin{itemize}
            \item $uv$ is not a bridge
            \item $uv$ lies on a path between odd cycles
        \end{itemize}
        Moreover, it fails to broadcast if and only if this is the first message sent along $uv$, $uv$ is a bridge, and the side of the cut containing $u$ does not contain an odd cycle.
    \end{theorem}
\begin{proof}[Proof of Single Message Failure Theorem]
    For the non-termination, this is an immediate consequence of theorem \ref{thm:balance} as the dropping of the message must produce an imbalance.\\
    For correctness, let $S_u$ be the configuration of Amnesiac Flooding on $G=(V,E)$ in the first round of broadcast from $u$. 
    A node $v$ will become informed via the shortest message path out of any message from $S$ ending at $v$, since this is the shortest path it cannot cross the same edge twice and so removing the second use of an edge will not harm correctness. 
    On the other hand if $uv$ is not a bridge, then the state reached by dropping a message on $uv$ the configuration becomes imbalanced and so all nodes will receive messages infinitely often. 
    Finally, if the message $(u,v)$ is dropped the first time $uv$ is used and $uv$ is a bridge, either there is an odd-cycle on $u$'s side of the cut or there isn't. 
    If there isn't then $u$ will be activated only once, as the component the messages can access is bipartite. 
    Otherwise, $u$ will be activated a second time and so the message will cross $uv$ and flood the other half as normal.
\end{proof}
We now turn to the result for Uni-directional link failure.
We will first need the following lemma.
\begin{lemma}\label{lemma: directed}
    Let $G=(V,E,A)$ be a simple-mixed communication graph such that $(u,v) \in A \implies (v,u) \notin A$, $C$ a cycle on $G$ containing at least one arc from $A$ and $S\subseteq V^2$ a message configuration on $G$. If $S_C$ contains more messages travelling in the same direction as an arc from $A$ than counter to it, $S_C$ is not a terminating configuration of Amnesiac Flooding on $G$.
\end{lemma}
\begin{proof}
    Since $C$ is a cycle on $G$, all arcs are oriented in the same direction with respect to $C$. 
    The number of messages on $C$ can change for three reasons: either messages are added due to external activation, messages are removed as a result of a collision between two messages, or a message reaches an arc oriented in the opposite direction. 
    The first creates either two messages, one travelling in each direction, or only a single message travelling in the same direction as the arcs. 
    The second removes one message travelling in both directions. 
    The third can remove only messages travelling counter to the arcs. 
    Therefore, there is no way for the number of messages travelling with the arcs to decrease relative to the number travelling counter to the arcs. 
    Thus, if there are more travelling with the arcs, the protocol can never reach termination.
\end{proof}
We can now prove the result for the uni-directional link failure case.
    \begin{theorem}[Uni-directional link failure]
    \label{thm: directed}
        For any graph $G=(V,E)$ and any initiator set $I\subsetneq V$ there exists an edge $e \in E$ such that a uni-directional link failure at $e$ will cause Amnesiac Flooding to either fail to broadcast or fail to terminate when initiated from $I$ on $G$. Furthermore, for any non-empty set of uni-directional link failures there exists $v \in V$ such that, when Amnesiac Flooding is initiated at $v$, it will either fail to broadcast or fail to terminate.
    \end{theorem}
\begin{proof}[Proof of Uni-directional Link Failure Theorem]
    For the first claim, there are two cases either the graph contains a bridge or it does not. 
    If the graph contains a bridge we take $e$ to be a bridge edge and orient it towards the broadcasting node. 
    This prevents broadcast as it disconnects the graph.\\
    On the other hand, if the graph does not contain a bridge it must contain a cycle. 
    We take the closest node $v$ of the cycle to our initiating node $u$ and orient one of its cycle edges $(w,v)$ towards it (please note that $v$ may in fact be $u$ or even non-unique).
    As $v$ is the closest node of the cycle to $u$ and nodes are always first informed over the shortest path from the initiating node, $v$ first receives a message it must do so from outside of the cycle.
    Furthermore, there can either be no other messages on the cycle when $v$ first receives the message.
    Then $v$ will attempt to send a message to both of its neighbours on the cycle, however $w$ will not receive the message as it is sent counter to the arc $(w,v)$.
    Any other node of the cycle being activated will cause the addition of one message in each direction, where as $v$ will only contribute a single message in the direction of the arc.
    Thus, there will be more messages travelling in the direction of the arc around the cycle and so by lemma \ref{lemma: directed} the protocol will not terminate.\\

    For the second claim, we observe that either the set of edge failures does not leave a strongly connected mixed-graph, or it contains a cycle with an arc. 
    In the former case there must be two nodes $u$ and $v$ such that no message from $u$ can reach $v$, so we take $u$ to be our initial node and the process fails to broadcast. 
    In the latter case, if the arc is $(u,v)$ we take the node $v$ to be our initial node. 
    In this case, immediately we have only a single message travelling in the direction of the arc around the cycle and none travelling counter to it. 
    Therefore, by application of lemma \ref{lemma: directed} the protocol will not terminate.
\end{proof}
Finally, we prove the Byzantine case.
    \begin{theorem}[Byzantine Failure]\label{thm: Byzantine}
        If Amnesiac Flooding on $G=(V,E)$ initiated from $I\subsetneq V$ experiences a \emph{weak Byzantine failure} at $J\subseteq V\setminus I$, then the adversary can force:
        \begin{itemize}
            \item Failure to broadcast if and only if $J$ contains a cut vertex set.
            \item Non-termination if and only if at least one member of $J$ lies on either a cycle or a path between odd-cycles.
        \end{itemize}
    \end{theorem}
\begin{proof}[Proof of Byzantine Failure Theorem]
    Since the Byzantine nodes are time bounded, for them to force the protocol not terminate they must selectively forward messages in order to force a non-terminating configuration at their final step. 
    By theorem \ref{thm:balance}, this requires them to force an imbalanced configuration.
    If the Byzantine set contains a node on a cycle or FEC, this is trivial as in a single step they can either send an additional message or fail to send a message, which will lead to imbalance.
    Conversely, if the Byzantine set does not contain such a node there is no way they can produce an imbalance on a cycle or FEC, as they can only externally activate (or deny external activations to) such a structure.
    But external activations do not affect the balance of a graph.
    Thus, the Byzantine agents can force non-termination if and only if the set contains a node on a cycle or FEC.\\

    The Byzantine agents can trivially prevent a broadcast if their set contains a cut node set, as Byzantine nodes on the cut can simply not send messages, effectively removing the nodes from the communication graph.
    However, if the set of nodes does not contain a cut node set, this strategy cannot disconnect the graph and so Amnesiac Flooding will still be successful.
    Assume for the sake of contradiction there exists a forwarding strategy such that the Byzantine agents can prevent broadcast.
    Let $B\subset V$ be the set of Byzantine nodes, $u \not \in B$ our initial node and $v\in V$ a node that will not be informed under this strategy.
    Since $B$ does not contain a cut node set there must be a shortest path $P$ from $u$ to $v$ in $G\setminus B$.
    As $v$ is not informed, the Byzantine nodes must prevent a message from travelling along from $u$ to $v$ along $P$.
    However, no member of $B$ lies on $P$ and so the only way to block a message on $P$ is to send a message in the other direction to collide with it.
    This can only be achieved by one of the Byzantine nodes $b_0$ sending a message that will reach a node $w$ on $P$ before the message from $u$ would.
    However, since $w$ has not yet received a message it must forward the message from $b_0$ in both directions along the path.
    Thus the message from $b_0$ will now inform $v$, and so must be prevented from reaching it.
    Iterating this argument, any message sent by the Byzantine nodes that blocks the message from $u$ along $P$ to $v$ will only lead to some other message reaching $v$ first.
    Thus, there is no strategy that can force the broadcast to be incorrect.
    
\end{proof}
\bibliographystyle{ACM-Reference-Format}
\bibliography{references}
\end{document}